\documentclass[11pt,amsfont,a4paper]{article}
\usepackage{latexsym,amssymb,amsmath,amsthm,color}
\usepackage{geometry}
\usepackage{url}
\usepackage{graphicx}
\usepackage[utf8]{inputenc}
\usepackage{authblk}
\usepackage{multirow}
\usepackage{comment}
\usepackage{tabularx}

\usepackage{tikz}
\usetikzlibrary{angles,calc,intersections,quotes,arrows.meta}
\usetikzlibrary{calc, arrows,backgrounds,positioning,fit}
\usetikzlibrary{decorations.markings}
\usetikzlibrary{shapes.geometric}
\usetikzlibrary{shapes.misc, fit}
\usetikzlibrary{decorations.pathreplacing}
\usepackage{pgfplots}
\tikzset{middlearrow/.style={
        decoration={markings,
            mark= at position 0.6 with {\arrow{#1}} ,
        },
        postaction={decorate}
    }
}

\theoremstyle{definition}
\newtheorem{definition}{Definition}
\newtheorem{construction}{Construction}
\newtheorem{example}[definition]{Example}

\theoremstyle{plain}
\newtheorem{theorem}{Theorem}
\newtheorem{proposition}[definition]{Proposition}
\newtheorem{lemma}[definition]{Lemma}
\newtheorem{remark}[definition]{Remark}
\newtheorem{corollary}[definition]{Corollary}

\newtheorem{claim}[definition]{Claim}

\title{Maximally recoverable codes with locality and availability}
\author{Umberto Mart{\'i}nez-Pe\~{n}as \thanks{umberto.martinez@uva.es}}
\affil{IMUVa-Mathematics Research Institute,\\University of Valladolid, Spain}
\author{V. Lalitha \thanks{lalitha.v@iiit.ac.in}}
\affil{Signal Processing and Communications Research Center (SPCRC),\\International Institute of Information Technology, Hyderabad, India}

\date{}

\begin{document}

\maketitle

\begin{abstract}
In this work, we introduce maximally recoverable codes with locality and availability. We consider locally repairable codes (LRCs) where certain subsets of $ t $ symbols belong each to $ N $ local repair sets, which are pairwise disjoint after removing the $ t $ symbols, and which are of size $ r+\delta-1 $ and can correct $ \delta-1 $ erasures locally. Classical LRCs with $ N $ disjoint repair sets and LRCs with $ N $-availability are recovered when setting $ t = 1 $ and $ t=\delta-1=1 $, respectively. Allowing $ t > 1 $ enables our codes to reduce the storage overhead for the same locality and availability. In this setting, we define maximally recoverable LRCs (MR-LRCs) as those that can correct any globally correctable erasure pattern given the locality and availability constraints. We then identify a large class of global erasure patterns that can be corrected by such MR-LRCs and prove that they are all the correctable patterns when $ t=1 $. We provide three explicit constructions of LRCs that can correct such erasure patterns (thus MR-LRCs for $ t=1 $), based on MSRD codes, each attaining the smallest finite-field sizes for some parameter regime. Finally, we extend the known lower bound on finite-field sizes from classical MR-LRCs to our setting (for any value of $ t $).

\textbf{Keywords:} Availability, distributed storage, linearized Reed-Solomon codes, locally repairable codes, maximally recoverable codes, partial MDS codes, sum-rank metric.
%
\end{abstract}

\section{Introduction} \label{sec intro}

Distributed Storage Systems (DSSs) have become crucial for numerous cloud-based services and other applications. However, the large scale of stored data in recent times has made disk failures a very common phenomenon, which translates into node erasures in the storage network. Unfortunately, data replication requires a very large storage overhead and is thus not feasible in practice. The traditional solution is to use an erasure code. Among them, MDS codes have maximum information rates (or minimum storage overheads). However, if an MDS code encodes $ k $ information symbols into $ n $ storage nodes, then it requires contacting and downloading $ k $ remaining nodes in order to recover from one single erasure in the network, which is the most common scenario. 

Locally Repairable Codes (LRCs) \cite{gopalan} appeared as a solution to this problem, since they can correct a single erasure by contacting a small number $ r \ll k $ (called locality) of surviving nodes (where $ k $ is the number of information symbols), while at the same time being able to correct a large number of erasures in catastrophic cases. There exist LRCs whose (global) minimum distances attain the upper bounds depending only on the code dimension and the locality restrictions \cite{tamo-barg}. However, such LRCs may not be able to correct all of the correctable erasure patterns given the locality constraints. Those LRCs who can are called maximally recoverable LRCs or MR-LRCs \cite{blaum-RAID, gopalan-MR} (sometimes they are also called partial MDS codes in the literature \cite{blaum-RAID}). Although they have the strongest erasure-correction capabilities among LRCs, there exist lower bounds on their field sizes indicating these have to grow superlinearly with the code length \cite{gopi}. Among different constructions, those from \cite{cai-field, gopi-field} asymptotically attain such lower bounds in some parameter regimes. 

Originally, LRCs were able to correct only one erasure locally \cite{gopalan}. In order to expand this local erasure-correction capability, two generalizations appeared. First, LRCs that can correct up to $ \delta - 1 $ erasures in each local repair set were introduced \cite{kamath}. In such LRCs, each symbol belongs to a local repair set, which we will assume are disjoint and of size $ r+\delta -1 $, as usual in the MR-LRC literature \cite{cai-field, gopi-field, gopi}. When up to $ \delta-1 $ erasures occur in one such local set, the erasures may be corrected by contacting at most $ r $ remaining nodes in that local set. A second approach is to provide each locally repairable symbol with availability \cite{availability, wang}. In LRCs with $ N $ availability, each information symbol belongs to $ N $ local sets which are pairwise disjoint after removing that one symbol. LRCs with availability also enable good parallel recovery, since the local sets can be accessed independently in order to access a target information symbol (since they are pairwise disjoint after removing that symbol). Parallel recovery is crucial in distributed storage systems involved in the management of hot data, where information is frequently accessed.  

A hybrid solution was proposed in \cite{cai-disjoint}, where an information symbol belongs to multiple local repair sets (pairwise disjoint after removing the symbol) that each can correct up to $ \delta - 1 $ erasures locally. However, both the works on LRCs with availability \cite{availability} and their generalization \cite{cai-disjoint} only studied them with respect to their (global) minimum distance, providing upper bounds and matching constructions. However, the fact that their minimum distance attains the upper bounds does not mean that such codes may correct every globally correctable erasure pattern (for the given locality constraints).

In this work, we introduce maximally recoverable LRCs with availability, which can correct all such globally correctable erasure patterns and have not yet been considered to the best of our knowledge. Furthermore, we consider a more general setting than that from \cite{cai-disjoint}. We consider LRCs where certain subsets of $ t $ symbols belong to $ N $ local repair sets (pairwise disjoint after removing such $ t $ symbols) and each of such local sets may correct up to $ \delta-1 $ erasures locally. LRCs with multiple disjoint repair sets as in \cite{cai-disjoint} are recovered from our definition by setting $ t=1 $, and the original LRCs with availability \cite{availability, wang} are recovered from our definition by setting $ t=\delta-1=1 $. Our LRCs with $ t>1 $ still provide locality $ r $, local correction $ \delta -1 $ and availability $ N $ to the $ t $ considered symbols but decrease the number of required local parities, thus reducing storage overhead compared to \cite{cai-disjoint, availability} (see Remark \ref{remark comparison cai parities}).

As was the case for classical MR-LRCs, since their correction capabilities are given by the locality constraints, our focus in this work will be to provide explicit constructions of MR-LRCs with availability that attain the smallest possible finite-field sizes. 

Our contributions are as follows. In Section \ref{sec problem description}, we describe the local erasure-correction setting and the LRCs that we consider in this work. We define MR-LRCs with availability as those that can globally correct any erasure pattern correctable by some LRC with availability for the same parameter values. We then identify a large family of global erasure patterns that such MR-LRCs with availability can correct, and we prove that they are all such erasure patterns when $ t=1 $. We conclude the section with Table \ref{table comparisons}, which summarizes the field sizes of the explicit constructions and the lower field size bound that we will obtain in the following sections. In Section \ref{appendix}, we revisit Maximum Sum-Rank Distance (MSRD) codes, which will be used for our constructions of MR-LRCs with availability. In the following three sections, we provide three constructions of LRCs with availability that can correct the previously identified erasure patterns, which in particular means they are MR-LRCs with availability in the case $ t=1 $. In Section \ref{sec const generator}, we provide a first explicit construction using generator matrices of MSRD codes. When $ t=1 $ and the MSRD code is a linearized Reed--Solomon code \cite{linearizedRS}, this construction coincides with \cite[Const. B]{cai-disjoint}. However, the maximal recovery property of such codes was not proven in \cite{cai-disjoint}. In Sections \ref{sec const parity-check} and \ref{sec const parity-check second}, we provide two new explicit constructions based on parity-check matrices of MSRD codes. As one can see in Table \ref{table comparisons}, each of these constructions attains a strictly smaller field size than the other two in some parameter regime (since the dominant term in the exponent depends on a different parameter), both when $ t=1 $ and when $ t>1 $. For realistic parameter regimes, Constructions \ref{construction 2} and \ref{construction 3} have smaller field sizes than Construction \ref{construction 1}. Finally, in Section \ref{sec lower bound}, we extend the lower bound on the field size of MR-LRCs from \cite{gopi} to our general setting.

\subsection*{Notation}

We fix a prime power $ q $ and denote by $ \mathbb{F}_q $ the finite field of size $ q $. We denote by $ \mathbb{F} $ an arbitrary field. For positive integers $ m $ and $ n $, we denote by $ \mathbb{F}^{m \times n} $ the set of matrices over $ \mathbb{F} $ of size $ m \times n $, and we denote $ \mathbb{F}^n = \mathbb{F}^{1 \times n} $. If $ m \leq n $, we define $ [m,n] = \{ m,m+1, \ldots, n \} $ and $ [n] = [1,n] $. Given a code $ \mathcal{C} \subseteq \mathbb{F}^n $, a matrix $ A \in \mathbb{F}^{m \times n} $, a vector $ \mathbf{c} \in \mathbb{F}^n $, and a subset $ R \subseteq [n] $, we denote by $ \mathcal{C}|_R \subseteq \mathbb{F}^{|R|} $, $ A|_R \in \mathbb{F}^{m \times |R|} $ and $ \mathbf{c}|_R \in \mathbb{F}^{|R|} $ their projections onto the coordinates in $ R $. For a matrix $ A \in \mathbb{F}^{m \times n} $ and a subset $ \mathcal{C} \subseteq \mathbb{F}^m $, we denote $ \mathcal{C} A = \{ \mathbf{c} A : \mathbf{c} \in \mathcal{C} \} $ and we denote by $ A^\intercal \in \mathbb{F}^{n \times m} $ the transpose of $ A $. Finally, for matrices $ A_i \in \mathbb{F}^{m_i \times n_i} $, for $ i \in [g] $, we denote by $ {\rm diag}(A_1, A_2, \ldots, A_g) \in \mathbb{F}^{m \times n} $ the block diagonal matrix with $ A_1, A_2, \ldots, A_g $ in the diagonal blocks, where $ m = m_1 + \cdots + m_g $ and $ n = n_1 + \cdots + n_g $.

\section{Problem Description and Main Results} \label{sec problem description}

In this section, we introduce the local erasure-correction setting that we will consider in this work and define MR-LRCs with availability as those LRCs with availability that can globally correct any erasure pattern that is correctable by some LRC with availability for the same parameter values. We then identify a large family of globally erasure patterns that are correctable by such MR-LRCs, showing that they are all the correctable patterns when $ t=1 $. We conclude with Table \ref{table comparisons}, which summarizes the field sizes obtained by our three explicit constructions from Sections \ref{sec const generator}, \ref{sec const parity-check} and \ref{sec const parity-check second}, together with the lower bound that we will obtain in Section \ref{sec lower bound}.

\begin{definition} \label{def LRC availability}
Consider positive integers $ r $, $ \delta $, $ t $, $ g $ and $ N $, where $ r $ will be the locality (or dimension of the local codes), $ \delta $ the local distance (hence the code will be able to correct $ \delta-1 $ erasures locally), $ g $ will be the number of disjoint blocks of local correction sets, $ t $ will be the number of symbols that have availability in each of the $ g $ blocks and $ N $ will be the availability (the number of local repair sets for such symbols). We define the code length  as
$$ n = g (t + N (r+\delta -1 -t)). $$
Assume that $ t \leq r $ and consider subsets $ T_i, R_{i,j} \subseteq [n] $, for $ i \in [g] $ and $ j \in [N] $, such that $ T_i \subseteq R_{i,j} $, $ |T_i| = t $, $ |R_{i,j}| = r+\delta - 1 $ (thus $ | R_{i,j} \setminus T_i | \geq \delta - 1 $ since $ t \leq r $) and
$$ R_{i,j} \cap R_{i,\ell} = T_i, $$
if $ j \neq \ell $, for $ j, \ell \in [N] $ and for $ i \in [g] $. In particular, $ T_i = \bigcap_{j=1}^N R_{i,j} $ if $ N \geq 2 $. Denote
$$ R_i = \bigcup_{j=1}^N R_{i,j} = T_i \cup \left( \bigcup_{j=1}^N (R_{i,j}\setminus T_i) \right) , $$
for $ i \in [g] $, and assume that
$$ [n] = \bigcup_{i=1}^g R_i \quad \textrm{and} \quad R_i \cap R_\ell  = \varnothing $$
(in particular, $ T_i \cap T_\ell = \varnothing $) if $ i \neq \ell $, for $ i,\ell \in [g] $. We say that a code $ \mathcal{C} \subseteq \mathbb{F}^n $ is a locally repairable code (LRC) of type $ (r,\delta,N,t) $ if 
$$ {\rm d}(\mathcal{C}|_{R_{i,j}}) \geq \delta, $$
for sets $ R_{i,j} $ as above, for all $ j \in [N] $ and $ i \in [g] $. We further say that $ \mathcal{C} $ has $ N $-availability if $ t \leq \delta - 1 $. If in addition, $ T = \bigcup_{i=1}^g T_i $ contains an information set (i.e., a set $ I \subseteq [n] $ such that the projection of the code onto the coordinates in $ I $ is one to one), then we say that $ \mathcal{C} $ has information $ N $-availability.
\end{definition}

Notice that, after a reordering of the coordinates in $ [n] $, we may assume that
$$ T_1 = [t], \quad R_{1,j} = [t] \cup [t + (j-1)(r+\delta -1 -t) + 1, t + j(r+ \delta - 1- t)], $$
for $ j \in [N] $, and 
\begin{equation}
\begin{split}
T_i & = (i-1)(t+N(r+\delta - 1-t)) + T_1, \\
R_{i,j} & = (i-1)(t+N(r+\delta - 1-t)) + R_{1,j},
\end{split}
\label{eq sets T_i R_ij}
\end{equation}
for $ j \in [N] $ and $ i \in [g] $. We may illustrate the interval $ [n] $ and these sets as follows:

\begin{center}
\resizebox{\textwidth}{!}{
\begin{tikzpicture}[scale=0.6]
    \draw (1,0)-- (26,0); 
    \foreach \x in {3,6,9,12,15,17,20,23} {
        \draw (\x,1) -- (\x,-0.2) node {};
    }    
    
    \draw (1,1) -- (1,-0.5) node[below] {$1$};
    \draw (26,1) -- (26,-0.5) node[below] {$n$};

	\draw (2,0) node[above] {$T_1$};
	\draw (4.5,0) node[above] {$R_{1,1} \setminus T_1$};
	\draw (7.5,0) node[above] {$\ldots$};
	\draw (10.5,0) node[above] {$R_{1,N} \setminus T_1$};
	\draw (13.5,0) node[above] {$\ldots$};
	\draw (16,0) node[above] {$T_g$};
	\draw (18.5,0) node[above] {$R_{g,1} \setminus T_g$};
	\draw (21.5,0) node[above] {$\ldots$};
	\draw (24.5,0) node[above] {$R_{g,N} \setminus T_g$};
	
	\draw (3,-0.5) node[below] {$t$};
	\draw (6,-0.5) node[below] {$r+\delta -1$};
	\draw (12,-0.5) node[below] {\tiny $\begin{array}{c}
	t +N(r+ \\
	\delta -1-t)
	\end{array}$};
	\draw (15,-0.5) node[below] {\tiny $\begin{array}{c}
	(g-1)(t + \\
	N(r+\delta -1-t))
	\end{array}$};
    
\end{tikzpicture}
}
\end{center}

\begin{remark}
According to this definition, for an LRC of type $ (r,\delta,N,t) $, the symbols in $ R_{i,j} $ have $ (r,\delta) $-locality according to \cite[Def. 1]{kamath}, since they can be recovered locally by using $ r $ symbols in $ R_{i,j} $ after erasing $ \delta - 1 $ other symbols from such a set. Furthermore, the symbols in $ T_i $ can be recovered after $ \delta - 1 $ erasures from any set $ R_{i,j} $, for $ j \in [N] $. When $ t = |T_i| \leq \delta - 1 $, then all the symbols in $ T_i $ can be recovered from any $ r $ symbols in any of the pairwise disjoint sets $ R_{i,j} \setminus T_i $, for $ j \in [N] $, and therefore the symbols in $ T = \bigcup_{i=1}^g T_i $ have $ N $-availability according to \cite[Def. 1]{availability} and \cite[Def. 1]{wang}. 
\end{remark}

\begin{remark} \label{remark comparison cai parities}
When $ t = 1 $, we recover LRCs with multiple disjoint repair sets as in \cite[Def. 4]{cai-disjoint}. However, when $ t \geq 2 $, notice that we are packing as many symbols in the intersections $ T_i = \bigcap_{j=1}^N R_{i,j} $ as possible, which still guarantees $ N $-availability for such symbols whenever $ t = |T_i| \leq \delta - 1 $. That is because in this case, if all $ t $ symbols in $ T_i $ are erased, we may still recover them using $ r $ symbols from each of the pairwise disjoint repair sets $ R_{i,j} \setminus T_i $, for $ j \in [N] $. Notice that, in this way, if we have $ k = gt $ information symbols, then we reduce the total number of local parities from
$$ kN(\delta -1) \quad \textrm{(the work \cite{cai-disjoint})} \qquad \textrm{to} \qquad gN(\delta-1) = \frac{kN(\delta-1)}{t} \quad \textrm{(this work)} . $$
If we further push the parameters to satisfy $ t = \delta -1 $, then the total number of local parities in this work is $ kN $, much lower than $ kN(\delta -1) $ (when $ \delta-1 \geq 2 $) as in \cite{cai-disjoint}. See also the example in Fig. \ref{fig local sets}.
\end{remark}

The assumption $ | R_{i,j} \setminus T_i | \geq \delta - 1 $, for $ j \in [N] $ and $ i \in [g] $, means that the pairwise disjoint sets $ R_{i,j} \setminus T_i $ each contains at least $ \delta - 1 $ local parities. After removing all of them, we obtain $ k $ information symbols plus a number $ h \geq 0 $ of global or heavy parities. Thus the number of heavy parities is defined as
$$ h = n - gN (\delta-1) - k = g(t + N(r-t)) - k \geq 0, $$
which in particular implies that $ k \leq g(t+N(r-t)) $. When $ N = 1 $, then $ h $ coincides with the number of heavy parities as considered in the maximally recoverable code literature \cite[Def. 2]{gopalan-MR}.

Before identifying globally correctable erasure patterns, we first describe those which are locally correctable. Notice that, for an erasure pattern $ \mathcal{E} \subseteq [n] $ to be correctable locally, we need to have $ | \mathcal{E} \cap R_{i,j_i} | \leq \delta - 1 $, for some $ j_i \in [N] $, for all $ i \in [g] $. Then, once we correct the erasures in $ \mathcal{E} \cap R_{i,j_i} $, there will be no erasures inside $ T_i $ (since $ T_i \subseteq R_{i,j_i} $). Hence, in order to continue correcting erasures in $ \mathcal{E} $ locally, we only need that $ (R_{i,\ell} \setminus T_i) \cap \mathcal{E} $ contains at most $ \delta - 1 $ elements, for $ \ell \in [N] \setminus \{ j_i \} $. 

\begin{definition} \label{def locally correctable erasures}
In the setting of Definition \ref{def LRC availability}, we say that an erasure pattern $ \mathcal{E} \subseteq [n] $ is locally correctable if, for every $ i \in [g] $, there exists $ j_i \in [N] $ such that
\begin{enumerate}
\item
$ | \mathcal{E} \cap R_{i,j_i} | \leq \delta - 1 $, and
\item
$ |(R_{i,\ell} \setminus T_i) \cap \mathcal{E}| \leq \delta - 1 $, for all $ \ell \in [N] \setminus \{ j_i \} $.
\end{enumerate}
We say that $ \mathcal{E} $ is maximal if equalities hold in the previous two items. In such a case, it must hold that
$$ |\mathcal{E}| = gN(\delta - 1). $$
See also Fig. \ref{fig local sets} for an example.
\end{definition}

\begin{figure*} [!h]
\begin{center}
\resizebox{\textwidth}{!}{
\begin{tabular}{c@{\extracolsep{1cm}}c}
\begin{tikzpicture}[
square/.style = {draw, rectangle, 
                 minimum size=\m, outer sep=0, inner sep=0, font=\small,
                 },
                        ]
\def\m{26pt}
\def\w{8}
\def\h{8}
    \pgfmathsetmacro\uw{int(\w/2)}
    \pgfmathsetmacro\uh{int(\h/2)}

\def\i{18}
  \foreach \x in {1,...,\w}
    \foreach \y in {1,...,\h}
       {    
			\node [square, fill=white]  (\x,\y) at (\x*\m + \i*\m,-\y*\m) {$ x_{\y,\x} $};        
       }
  \foreach \x in {1,...,\w}
    \foreach \y in {1,2,7,8}
       {    
			\node [square, fill=gray!50]  (\x,\y) at (\x*\m + \i*\m,-\y*\m) {$ x_{\y,\x} $};
       }
  \foreach \x in {1,...,\w}
    \foreach \y in {4,5}
       {    
			\node [square, fill=white]  (\x,\y) at (\x*\m + \i*\m,-\y*\m) {$ x_{\y,\x} $};
       }   
   
   \draw[thick, dotted] (0.4*\m + \i*\m, -0.4*\m) rectangle (1.4*\m + \i*\m,-5.6*\m); 
   \draw[thick, dashed] (0.6*\m + \i*\m, -3.4*\m) rectangle (1.6*\m + \i*\m,-8.6*\m); 
   \draw[thick, dashdotted] (0.3*\m + \i*\m, -3.3*\m) rectangle (8.6*\m + \i*\m,-5.7*\m); 
        
\end{tikzpicture}

\begin{tikzpicture}[
square/.style = {draw, rectangle, 
                 minimum size=\m, outer sep=0, inner sep=0, font=\small,
                 },
                        ]
\def\m{26pt}
\def\w{8}
\def\h{8}
    \pgfmathsetmacro\uw{int(\w/2)}
    \pgfmathsetmacro\uh{int(\h/2)}

\def\i{18}
  \foreach \x in {1,...,\w}
    \foreach \y in {1,...,\h}
       {    
			\node [square, fill=white]  (\x,\y) at (\x*\m + \i*\m,-\y*\m) {$ x_{\y,\x} $};        
       }
  \foreach \x in {1,...,\w}
    \foreach \y in {1,2,7,8}
       {    
			\node [square, fill=gray!50]  (\x,\y) at (\x*\m + \i*\m,-\y*\m) {$ x_{\y,\x} $};
       }
  \foreach \x in {1,...,\w}
    \foreach \y in {4,5}
       {    
			\node [square, fill=white]  (\x,\y) at (\x*\m + \i*\m,-\y*\m) {$ x_{\y,\x} $};
       }

   \node [square, fill=gray!50]  at (1*\m + \i*\m,-1*\m) {$ \spadesuit $};   
   \node [square, fill=white]  at (1*\m + \i*\m,-3*\m) {$ \spadesuit $}; 
   \node [square, fill=white]  at (1*\m + \i*\m,-4*\m) {$ \spadesuit $};   
   \node [square, fill=gray!50]  at (1*\m + \i*\m,-7*\m) {$ \spadesuit $}; 
   
   \node [square, fill=gray!50]  at (2*\m + \i*\m,-2*\m) {$ \spadesuit $};   
   \node [square, fill=white]  at (2*\m + \i*\m,-5*\m) {$ \spadesuit $}; 
   \node [square, fill=gray!50]  at (2*\m + \i*\m,-7*\m) {$ \spadesuit $};   
   \node [square, fill=gray!50]  at (2*\m + \i*\m,-8*\m) {$ \spadesuit $}; 
   
   \node [square, fill=white]  at (3*\m + \i*\m,-3*\m) {$ \spadesuit $};   
   \node [square, fill=white]  at (3*\m + \i*\m,-4*\m) {$ \spadesuit $}; 
   \node [square, fill=white]  at (3*\m + \i*\m,-6*\m) {$ \spadesuit $};   
   \node [square, fill=gray!50]  at (3*\m + \i*\m,-8*\m) {$ \spadesuit $}; 
   
   \node [square, fill=gray!50]  at (4*\m + \i*\m,-1*\m) {$ \spadesuit $};   
   \node [square, fill=gray!50]  at (4*\m + \i*\m,-2*\m) {$ \spadesuit $}; 
   \node [square, fill=white]  at (4*\m + \i*\m,-6*\m) {$ \spadesuit $};   
   \node [square, fill=gray!50]  at (4*\m + \i*\m,-7*\m) {$ \spadesuit $}; 
   
   \node [square, fill=white]  at (5*\m + \i*\m,-4*\m) {$ \spadesuit $};   
   \node [square, fill=white]  at (5*\m + \i*\m,-5*\m) {$ \spadesuit $}; 
   \node [square, fill=white]  at (5*\m + \i*\m,-6*\m) {$ \spadesuit $};   
   \node [square, fill=gray!50]  at (5*\m + \i*\m,-7*\m) {$ \spadesuit $}; 
   
   \node [square, fill=white]  at (6*\m + \i*\m,-3*\m) {$ \spadesuit $};   
   \node [square, fill=white]  at (6*\m + \i*\m,-4*\m) {$ \spadesuit $}; 
   \node [square, fill=gray!50]  at (6*\m + \i*\m,-1*\m) {$ \spadesuit $};   
   \node [square, fill=gray!50]  at (6*\m + \i*\m,-7*\m) {$ \spadesuit $}; 
   
   \node [square, fill=gray!50]  at (7*\m + \i*\m,-1*\m) {$ \spadesuit $};   
   \node [square, fill=gray!50]  at (7*\m + \i*\m,-2*\m) {$ \spadesuit $}; 
   \node [square, fill=gray!50]  at (7*\m + \i*\m,-7*\m) {$ \spadesuit $};   
   \node [square, fill=gray!50]  at (7*\m + \i*\m,-8*\m) {$ \spadesuit $}; 
   
   \node [square, fill=gray!50]  at (8*\m + \i*\m,-2*\m) {$ \spadesuit $};   
   \node [square, fill=white]  at (8*\m + \i*\m,-3*\m) {$ \spadesuit $}; 
   \node [square, fill=white]  at (8*\m + \i*\m,-4*\m) {$ \spadesuit $};   
   \node [square, fill=white]  at (8*\m + \i*\m,-5*\m) {$ \spadesuit $};

   \draw[thick, dashdotted] (0.4*\m + \i*\m, -3.4*\m) rectangle (8.6*\m + \i*\m,-5.6*\m);         
     
\end{tikzpicture}

\end{tabular}
}
\end{center}

\caption{On the left, illustration of a codeword in a code as in Definition \ref{def LRC availability} for $ r = 3 $, $ \delta = 3 $, $ t = 2 $, $ g = 8 $ and $ N = 2 $. The symbols are arranged in a matrix, where the $ i $th column corresponds to $ R_i $, the first $ 5 $ symbols are those in $ R_{i,1} $ (dotted line) and the last $ 5 $ symbols are those in $ R_{i,2} $ (dashed line). The symbols in the $ 4 $th and $ 5 $th rows (dash-dotted line) are those in $ T = \bigcup_{i=1}^g T_i $ (notice that $ T_i = R_{i,1} \cap R_{i,2} $). For example, if we index the symbols as in Definition \ref{def LRC availability}, we would have that $ x_{4,1}, x_{5,1}, x_{1,1}, x_{2,1}, x_{3,1}, x_{6,1}, x_{7,1}, x_{8,1} $ would be the first $ 8 $ symbols of the codeword, corresponding to $ R_1 = [8] $. Local parities are depicted in grey. If all information symbols correspond to the set $ T = \bigcup_{i=1}^g T_i $, then $ k = gt = 16 $, and then the total number of local parities is $ kN = 32 $, in contrast with $ kN(\delta-1) = 64 $ as in \cite{cai-disjoint}, while we still guarantee $ N $ disjoint repair sets that can recover from $ \delta -1 $ erasures each, for each information symbol (see Remark \ref{remark comparison cai parities}). On the right, an example of a maximal locally correctable erasure pattern as in Definition \ref{def locally correctable erasures}. If the code is maximally recoverable, then the restricted code obtained after removing such coordinates is MDS of length $ k+h = 32 $ and dimension $ k = 16 $.}
\label{fig local sets}
\end{figure*}

As is usual, a maximally recoverable LRC in this scenario must (globally) correct any erasure pattern that can be corrected by some other LRC of the same parameters.

\begin{definition} \label{def mr}
An LRC $ \mathcal{C} \subseteq \mathbb{F}^n $ of type $ (r,\delta, N,t) $ with information $ N $-availability as in Definition \ref{def LRC availability} is called maximally recoverable (or an MR-LRC) if it can correct any erasure pattern $ \mathcal{E} \subseteq [n] $ that is correctable by some LRC of type $ (r,\delta,N,t) $ with information $ N $-availability. 
\end{definition}

\begin{remark}
Observe that classical MR-LRCs \cite{blaum-RAID, gopalan-MR} are recovered from Definition \ref{def mr} simply by setting $ N = 1 $, in which case the value of $ t $ is irrelevant.
\end{remark}

We now identify a family of erasure patterns that maximally recoverable codes can globally correct. We will prove that, in the case $ t=1 $, they are in fact all of the erasure patterns correctable by maximally recoverable codes.

\begin{definition} \label{def correctable patterns}
Consider parameters as in Definition \ref{def LRC availability}. We say an erasure pattern $ \mathcal{E} \subseteq [n] $ is simple if
\begin{enumerate}
\item
$ \mathcal{E} = \mathcal{E}_1 \cup \mathcal{E}_2 $, where $ \mathcal{E}_1 \cap \mathcal{E}_2 = \varnothing $,
\item
$ \mathcal{E}_2 $ is a locally correctable pattern (Definition \ref{def locally correctable erasures}), and
\item
$ |\mathcal{E}_1| \leq h = g(t + N(r-t)) - k $.
\end{enumerate}
\end{definition}

\begin{remark} \label{remark simple patterns if MDS}
Let $ \mathcal{C} \subseteq \mathbb{F}^n $ be an LRC of type $ (r,\delta, N,t) $ with information $ N $-availability as in Definition \ref{def LRC availability}. Then $ \mathcal{C} $ can correct all simple erasure patterns as in Definition \ref{def correctable patterns} if, and only if, for any erasure pattern $ \mathcal{E} \subseteq [n] $ that is locally correctable (Definition \ref{def locally correctable erasures}) and of maximum size $ |\mathcal{E}| = gN(\delta - 1) $, the restricted code $ \mathcal{C}_{glob}|_{\overline{\mathcal{E}}} $ is MDS, where $ \overline{\mathcal{E}} = [n] \setminus \mathcal{E} $.
\end{remark}

We now show that these erasure patterns are correctable by maximally recoverable codes and, when $ t = 1 $, they constitute all of the correctable erasure patterns.

\begin{theorem} \label{thm:correctable}
An MR-LRC of type $ (r,\delta, N,t) $ with information $ N $-availability can correct all of the erasure patterns given in Definition \ref{def correctable patterns}. Moreover, when $ t = 1 $, these are all of the erasure patterns that such a code can correct.
\end{theorem}
\begin{proof}
The erasure patterns from Definition \ref{def mr} are correctable by the constructions that we will give in Sections \ref{sec const generator}, \ref{sec const parity-check} and \ref{sec const parity-check second}. Therefore, the first part follows from Definition \ref{def mr} (i.e., there exists some LRC for the corresponding parameter values that can correct such erasure patterns).

We set $ t=1 $ for the remainder of the proof. Let $ \mathcal{E} \subseteq [n] $ be an erasure pattern that can be corrected by some LRC $ \mathcal{C} \subseteq \mathbb{F}^n $ of type $ (r,\delta,N,t) $ with information $ N $-availability. We only need to prove that $ \mathcal{E} $ is simple as in Definition \ref{def correctable patterns}.

We will modify $ \mathcal{E} $ successively in a way that, in each iteration, the next pattern is correctable and simple if, and only if, so was the previous pattern. We will arrive at a correctable pattern that is easy to see that is simple, and we will be done (i.e., the original pattern $ \mathcal{E} $ must be simple too).

Let $ i \in [g] $ and consider the case that $ |\mathcal{E} \cap R_i| \leq N(\delta - 1) $. There must exist $ j_i \in [N] $ such that $ |\mathcal{E} \cap R_{i,j_i}| \leq \delta - 1 $. This is because if $ |\mathcal{E} \cap R_{i,j}| \geq \delta $, for all $ j \in [N] $, then $ |\mathcal{E} \cap R_i | \geq N(\delta - 1) + 1 $, since $ |R_{i,j} \cap R_{i,\ell}| = |T_i| = t = 1 $ if $ i \neq \ell $ (here we use the assumption $ t=1 $). Since $ \mathcal{C} $ is an LRC, we can remove the entire set $ \mathcal{E} \cap R_{i,j_i} $ from $ \mathcal{E} $ (i.e., locally correct it) and replace it by $ |\mathcal{E} \cap R_{i,j_i}| $ elements inside $ R_{i,j_i} \setminus T_i $ (recall that $ |R_{i,j_i} \setminus T_i| \geq \delta - 1 $). In this way, we may assume that $ \mathcal{E} \cap T_i = \varnothing $. Now, if $ | \mathcal{E} \cap (R_{i,j} \setminus T_i) | < \delta - 1 $, for some $ j \in [N] $, then we may add elements to $ \mathcal{E} \cap (R_{i,j} \setminus T_i) $ until $ | \mathcal{E} \cap (R_{i,j} \setminus T_i) | = \delta - 1 $, since in this case $ \mathcal{E} \cap (R_{i,j} \setminus T_i) $ is locally correctable and $ |R_{i,j} \setminus T_i| \geq \delta - 1 $. Hence we may assume that $ | \mathcal{E} \cap (R_{i,j} \setminus T_i) | \geq \delta - 1 $, for all $ j \in [N] $. 

In other words, we may assume that $ |\mathcal{E} \cap R_i| \geq N(\delta - 1) $, i.e., 
$$ | \mathcal{E} \cap R_i | = N(\delta - 1) + e_i, $$
where $ e_i \geq 0 $, for all $ i \in [g] $, and where $ | \mathcal{E} \cap (R_{i,j} \setminus T_i)| = \delta - 1 $ for all $ j \in [N] $ if $ e_i=0 $. Note that, since $ \mathcal{E} $ is correctable, then it must hold that
$$ \sum_{i=1}^g e_i \leq h. $$

Now let $ i \in [g] $ be such that $ e_i > 0 $. We distinguish two cases. First, consider the case $ \mathcal{E} \cap T_i = \varnothing $. As before, since $ \mathcal{C} $ is locally correctable, we may assume that $ |\mathcal{E} \cap (R_{i,j} \setminus T_i) | \geq \delta - 1 $, for all $ j \in [N] $. In this case, we set
$$ |\mathcal{E} \cap (R_{i,j} \setminus T_i)| = \delta - 1 + e_{i,j}, $$
where $ e_{i,j} \geq 0 $, for all $ j \in [N] $. The other possible case is $ T_i \subseteq \mathcal{E} $ since $ |T_i| = t = 1 $ (here we also use the assumption $ t=1 $). Once again, if $ |\mathcal{E} \cap (R_{i,j}\setminus T_i)| \leq \delta -2 $, for some $ j \in [N] $, then we may again add elements to $ \mathcal{E} \cap (R_{i,j}\setminus T_i) $ until $ |\mathcal{E} \cap (R_{i,j}\setminus T_i)| = \delta -2 $, which means $ |\mathcal{E} \cap R_{i,j}| = \delta - 1 $. In this second case, we set 
$$ |\mathcal{E} \cap (R_{i,j} \setminus T_i)| = \left\lbrace \begin{array}{ll}
\delta - 1 + e_{i,1} - 1 & \textrm{ if } j = 1, \\
\delta - 1 + e_{i,j} & \textrm{ if } j > 1,
\end{array} \right. $$
where $ e_{i,j} \geq 0 $ and $ e_{i,1} \geq 1 $, for all $ j \in [N] $ (in a way, we ``assign'' the element in $ T_i \subseteq \mathcal{E} $ to $ \mathcal{E} \cap R_{i,1} $ and no other set $ \mathcal{E} \cap R_{i,j} $). Note that, as we modified $ \mathcal{E} $, the values of $ e_i $ and $ e_{i,j} $, for $ i \in [g] $ and $ j \in [N] $, also changed.

In all cases, we have that 
$$ e_i = \sum_{j=1}^N e_{i,j}, $$
for all $ i \in [g] $. Next, since $ \mathcal{E} $ is correctable by a code of length $ n $ and dimension $ k $, and $ n-k = h + gN(\delta - 1) $, then 
$$ |\mathcal{E}| \leq h + gN(\delta- 1). $$
Since $ |\mathcal{E} \cap R_i| = N(\delta - 1) + e_i $, for all $ i \in [g] $, then
$$ \sum_{i=1}^g e_i \leq h. $$
Now, since $ |\mathcal{E} \cap (R_{i,j} \setminus T_i)| \geq \delta - 1 $, we may define $ \mathcal{E}_2 \subseteq \mathcal{E} $ by choosing any $ \delta - 1 $ elements in $ \mathcal{E} \cap (R_{i,j} \setminus T_i) $, for all $ i \in [g] $ and all $ j \in [N] $. Setting $ \mathcal{E}_1 = \mathcal{E} \setminus \mathcal{E}_2 $, we deduce that
$$ |\mathcal{E}_1| = |\mathcal{E}| - |\mathcal{E}_2| = \sum_{i=1}^g e_i \leq h. $$
It is obvious that $ \mathcal{E}_2 $ is locally correctable since $ \mathcal{E}_2 \cap T_i = \varnothing $ and $ |\mathcal{E}_2 \cap (R_{i,j} \setminus T_i)| = \delta - 1 $, for all $ i \in [g] $ and all $ j \in [N] $. Therefore, $ \mathcal{E} $ is simple by Definition \ref{def correctable patterns}. 

Finally, the original erasure pattern must be simple too, since every time we modified it, if the next pattern was simple, then so was the previous one, hence we are done. 
\end{proof}

In the following counterexample, we show that, for $ t=2 $, there exist LRCs with availability that can correct erasure patterns that are not simple as in Definition \ref{def correctable patterns} (i.e., the second part of Theorem \ref{thm:correctable} does not hold for $ t=2 $). 

\begin{example}
Set $ t = \delta - 1 = r = k = N = 2 $, $ g=1 $, $ h=0 $, $ T_1 = \{ 1,2 \} $, $ R_{1,1} = \{ 1,2,3,4 \} $ and $ R_{1,2} = \{ 1,2,5,6 \} $ (thus $ n = 6 $ and $ R_1 = [n] = \{ 1,2,3,4,5,6 \} $). Consider the linear code $ \mathcal{C} \subseteq \mathbb{F}_4^6 $ with generator matrix
$$ G = \left( \begin{array}{cccccccc}
1 & 0 & \alpha & \alpha + 1 & \alpha & \alpha+1 \\
0 & 1 & \alpha+1 & \alpha & \alpha+1 & \alpha
\end{array} \right), $$
where $ \mathbb{F}_4 = \{ 0,1,\alpha,\alpha+1 \} $ and $ \alpha^2 = \alpha+1 $. Then $ \mathcal{C} $ is an LRC of type $ (r,\delta,N,t) $ with $ N $-availability. Now, it can also correct the erasure pattern $ \mathcal{E} = \{ 1,2,3,6 \} $, since the following matrix is invertible 
$$ \left( \begin{array}{cc}
\alpha+1 & \alpha \\
\alpha & \alpha+1
\end{array} \right). $$
However, this pattern is not locally correctable, hence it is not simple according to Definition \ref{def correctable patterns} (due to the fact that $ h=0 $).
\end{example}

\begin{table*} [!t] \tiny
\centering
\begin{tabularx}{\textwidth}{>{\hsize=0.4\hsize}X||>{\hsize=0.4\hsize}X|X|>{\hsize=0.6\hsize}X|>{\hsize=0.6\hsize}X}
\hline
&&\\[-0.8em]
Constructions & Explicit code & Upper bound on field size $ q^m $ & Restrictions & Code it generalizes \\[0.3em]
\hline\hline
&&&&\\[-0.8em]
Construction \ref{construction 1} & Corollary \ref{cor mr-lrc construction gen matrix lrs} & $ \mathcal{O}(\max \{ g+1, r+\delta-1 \})^{t+N(r-t)} $ & $ t \leq \min \{ \delta - 1, r \} $ & $ \begin{array}{c}
 \textrm{\cite[Const. 1]{universal-lrc},} \\
 \textrm{\cite[Const. B]{cai-disjoint}}
\end{array} $ \\[0.3em]
\hline 
&&&&\\[-0.8em]
Construction \ref{construction 2} & Corollary \ref{cor mr-lrc construction parity matrix lrs} & $ \mathcal{O}(\max \{ g+1, r+\delta-1 \})^{hN} $ & $ \begin{array}{c}
t \leq \min \{ \delta - 1, r \}, \\
h \leq \min \{ r, \\
g(t + N(r-t)) \}
\end{array} $ & $ \begin{array}{c}
 \textrm{\cite[Sec. III]{cai-field},} \\
 \textrm{\cite[Sec. III-A]{gopi-field}}
\end{array} $ \\[0.3em]
\hline 
&&&&\\[-0.8em]
Construction \ref{construction 3} & Corollary \ref{cor mr-lrc construction parity matrix lrs 2} & $ \mathcal{O}\left( \frac{n}{g} - 1 \right)^{gN(\delta-1)+h} $ & $ \begin{array}{c}
t \leq \min \{ \delta - 1, r \}, \\
h \leq g(t + N(r-t))
\end{array} $  & \cite[Sec. V-A]{gopalan-MR}  \\[0.3em]
\hline 
\hline 
&&&&\\[-0.8em]
Lower bounds &  & Lower bound on field size $ q^m $ & Restrictions & Bound it generalizes \\[0.3em]
\hline 
\hline
&&&&\\[-0.8em]
Corollary \ref{cor lower bounds field h <= g} &  & $ \Omega_{h,\delta,N} \left( gt \cdot r^{\min \{ N(\delta-1),N \lfloor \frac{h-2}{N} \rfloor \}} \right) $ & $ h \leq g $ & \cite[Cor. 1]{gopi}  \\[0.3em]
\hline 
\end{tabularx}
\caption{Table summarizing the field sizes of Constructions \ref{construction 1}, \ref{construction 2} and \ref{construction 3}, together with the lower bound on field sizes from Corollary \ref{cor lower bounds field h <= g}. }
\label{table comparisons}
\end{table*}

In the following three sections, we will provide three explicit constructions of LRCs with availability as in Definition \ref{def LRC availability} that can correct any simple erasure pattern as in Definition \ref{def correctable patterns}, which in particular means that they are MR-LRCs with availability when $ t=1 $ by Theorem \ref{thm:correctable}. We also give a general lower bound on their field sizes (valid for any $ t $). We collect such results in Table \ref{table comparisons}, emphasizing the field sizes of our constructions. As a quick guide, the field size $ q $ is generally used for local correction (and is generally very small and not studied), whereas the larger field size $ q^m $ is used for global correction, where $ m $ is the degree of the field extension. Our objective is to make the field size $ q^m $ as small as possible. 
Each of the three constructions attains smaller field sizes $ q^m $ than the other two for some parameter regime, since the exponent (which is the dominant term) depends on a different parameter in each construction. Notice however that, in realistic scenarios, $ r $ grows with the code length, whereas $ \delta $, $ h $, $ g $ and $ N $ are constant. Therefore, Constructions \ref{construction 2} and \ref{construction 3} have generally smaller field sizes than Construction \ref{construction 1}.

\section{Maximum Sum-Rank Distance (MSRD) Codes} \label{appendix}

In this section, we revisit some basic facts regarding MSRD codes, which will be used in the constructions in the following sections. For more details, we refer the reader to \cite{fnt}. The sum-rank metric \cite{spacetime-kumar, multishot} is defined as follows.

\begin{definition} [\cite{spacetime-kumar, multishot}]
Given $ \mathbf{c} = (\mathbf{c}_1, \mathbf{c}_2, \ldots, \mathbf{c}_g) \in \mathbb{F}_{q^m}^{gr} $, where $ \mathbf{c}_i = (c_{i,1}, c_{i,2}, \ldots, c_{i,r}) \in \mathbb{F}_{q^m}^r $, for $ i \in [g] $, we define its sum-rank weight over $ \mathbb{F}_q $ for the length partition $ (g,r) $ as
$$ {\rm wt}_{SR}(\mathbf{c}) = \sum_{i=1}^g \dim_{\mathbb{F}_q} (\langle c_{i,1}, c_{i,2}, \ldots, c_{i,r} \rangle_{\mathbb{F}_q}), $$
where $ \dim_{\mathbb{F}_q} $ and $ \langle \cdot \rangle_{\mathbb{F}_q} $ denote dimension and linear span over $ \mathbb{F}_q $, respectively. Given $ \mathbf{c}, \mathbf{d} \in \mathbb{F}_{q^m}^{gr} $, we define their sum-rank distance (over $ \mathbb{F}_q $ for the length partition $ (g,r) $) as
$$ {\rm d}_{SR}(\mathbf{c},\mathbf{d}) = {\rm wt}_{SR}(\mathbf{c} - \mathbf{d}). $$
Finally, given a code $ \mathcal{C} \subseteq \mathbb{F}_{q^m}^{gr} $, we define its minimum sum-rank distance as
$$ {\rm d}_{SR}(\mathcal{C}) = \min \{ {\rm d}_{SR}(\mathbf{c},\mathbf{d}) : \mathbf{c},\mathbf{d} \in \mathcal{C}, \mathbf{c} \neq \mathbf{d} \}. $$
\end{definition}

As in the case of the Hamming metric, a Singleton bound holds for the sum-rank metric \cite[Cor. 2]{universal-lrc}.

\begin{proposition} [\cite{universal-lrc}]
For a code $ \mathcal{C} \subseteq \mathbb{F}_{q^m}^n $, $ n = gr $, if $ k = \log_{q^m} |\mathcal{C}| $, then
\begin{equation}
{\rm d}_{SR}(\mathcal{C}) \leq n - k + 1.
\label{eq singleton}
\end{equation}
\end{proposition}

\begin{definition} [\cite{universal-lrc}] \label{def msrd}
We say that the code $ \mathcal{C} \subseteq \mathbb{F}_{q^m}^{gr} $ is maximum sum-rank distance (MSRD) if it attains the bound (\ref{eq singleton}).
\end{definition}

\begin{remark} \label{remark condition m geq r for msrd}
By \cite[Cor. 3]{universal-lrc}, if an MSRD code exists in $ \mathbb{F}_{q^m}^{gr} $, then $ m \geq r $.
\end{remark}

For our purposes, we will need the following characterization of MSRD codes from \cite[Cor. 2]{universal-lrc}.

\begin{lemma} [\cite{universal-lrc}] \label{lemma msrd charact mds diag}
A code $ \mathcal{C} \subseteq \mathbb{F}_{q^m}^{gr} $ is MSRD (over $ \mathbb{F}_q $ for the length partition $ (g,r) $) if, and only if, for all invertible matrices $ A_1, A_2, \ldots, A_g \in \mathbb{F}_q^{r \times r} $, the code $ \mathcal{C} {\rm diag}(A_1, A_2, \ldots, A_g) $ $ \subseteq \mathbb{F}_{q^m}^{gr} $ is MDS.
\end{lemma}

An explicit construction of MSRD codes with small field sizes is given by linearized Reed--Solomon codes \cite[Def. 31]{linearizedRS}. See also \cite[Ch. 2, Def. 2.14]{fnt}.

\begin{definition} [\cite{linearizedRS}] \label{def lrs codes}
Let $ q > g $ be a power of a prime and let $ m \geq r $. Let $ a_1,a_2, \ldots, a_g \in \mathbb{F}_{q^m}^* $ be such that $ N_{\mathbb{F}_{q^m}/\mathbb{F}_q}(a_i) \neq N_{\mathbb{F}_{q^m}/\mathbb{F}_q}(a_j) $ if $ i \neq j $ (where $ N_{\mathbb{F}_{q^m}/\mathbb{F}_q} $ denotes the norm of $ \mathbb{F}_{q^m} $ over $ \mathbb{F}_q $ \cite[Ch. 2, Sec. 3]{lidl}). Let $ \beta_1, \beta_2, \ldots, \beta_r \in \mathbb{F}_{q^m} $ be $ \mathbb{F}_q $-linearly independent.

Set $ \mathbf{a} = (a_1,a_2, \ldots,a_g) $ and $ \boldsymbol\beta = (\beta_1, \beta_2, \ldots, \beta_m) $, and set $ n = gr $. For $ k \in [gr] $, we define the $ k $-dimensional linearized Reed--Solomon code in $ \mathbb{F}_{q^m}^{gr} $ over $ (\mathbf{a}, \boldsymbol\beta) $ as that with generator matrix
$$ G = (G_1|G_2| \ldots |G_g) \in \mathbb{F}_{q^m}^{k \times (gr)}, $$
where
$$ G_i = \left( \begin{array}{cccc}
\beta_1 & \beta_2 & \ldots & \beta_r \\
 \beta_1^q a_i & \beta_2^q a_i & \ldots & \beta_r^q a_i \\
 \beta_1^{q^2} a_i^{q+1} & \beta_2^{q^2} a_i^{q+1} & \ldots & \beta_r^{q^2} a_i^{q+1} \\
\vdots & \vdots & \ddots & \vdots \\
 \beta_1^{q^{k-1}} a_i^{\frac{q^{k-1}-1}{q-1}} & \beta_2^{q^{k-1}} a_i^{\frac{q^{k-1}-1}{q-1}} & \ldots & \beta_r a_i^{\frac{q^{k-1}-1}{q-1}}
\end{array} \right) \in \mathbb{F}_{q^m}^{k \times r}, $$
for $ i \in [g] $. 
\end{definition}

We notice that, in this construction, we may choose a different $ r $-size set $ \{ \beta_1, \beta_2, \ldots, \beta_r \} $ $ \subseteq \mathbb{F}_{q^m} $ of $ \mathbb{F}_q $-linearly independent elements for different block indices $ i \in [g] $. Moreover, such sets may have different sizes $ r_i $ for different indices $ i \in [g] $ (in the definition of the sum-rank metric, we would partition the code length as $ r_1+r_2+ \cdots +r_g $). However, we omit this case for simplicity in the notation. See also \cite{fnt}.

The bounds $ q > g $ and $ m \geq r $ are tight, meaning that we can choose $ q = \mathcal{O}(g) $ and $ m = r $. In that case, the field sizes of linearized Reed--Solomon codes are of the form $ q^m = \mathcal{O}(g)^r $. The following result is \cite[Th. 4]{linearizedRS}. See also \cite[Ch. 2, Th. 2.20]{fnt}.

\begin{theorem} [\cite{linearizedRS}] \label{th lrs codes are msrd}
With notation and hypotheses as in Definition \ref{def lrs codes}, a linearized Reed--Solomon code is MSRD over $ \mathbb{F}_q $ for the length partition $ (g,r) $.
\end{theorem}

\section{A Construction Using Generator Matrices} \label{sec const generator}

In this section, we will present a first construction, which is based on generator matrices of maximum sum-rank distance (MSRD) codes (see Section \ref{appendix}). The construction extends both \cite[Const. 1]{universal-lrc} and \cite[Const. B]{cai-disjoint}. For $ t=1 $ and choosing linearized Reed--Solomon codes \cite{linearizedRS} as MSRD codes, this construction essentially becomes \cite[Const. B]{cai-disjoint}. However, the maximal recovery property of such codes was not proven in that work, whereas we will prove it. 
%

\begin{construction} \label{construction 1}
Fix positive integers $ r $, $ \delta $, $ t $, $ g $ and $ N $, with $ t \leq \min \{ \delta - 1 , r \} $. Define the code lengths
$$ n_o = g(t + N(r-t)) \quad \textrm{and} \quad n = g(t + N(r+\delta - 1-t)). $$
Next choose a dimension $ k \in [n_o] $ and:
\begin{enumerate}
\item
\textit{Outer code}: A $ k $-dimensional $ \mathbb{F}_{q^m} $-linear code $ \mathcal{C}_{out} \subseteq \mathbb{F}_{q^m}^{n_o} $ that is MSRD over $ \mathbb{F}_q $ for the length partition $ (g, t+N(r-t)) $. See Definition \ref{def msrd}.
\item
\textit{Local codes}: An $ (r+\delta-1, r) $ MDS code $ \mathcal{C}_{loc} \subseteq \mathbb{F}_q^{r + \delta - 1} $, linear over $ \mathbb{F}_q $. Let $ A \in \mathbb{F}_q^{r \times (r + \delta - 1)} $ be a generator matrix of $ \mathcal{C}_{loc} $ of the form 
$$ A = \left( \begin{array}{c|c}
I_t & B \\
\hline
0_{(r-t) \times t} & C
\end{array} \right) , $$ 
where $ I_t $ denotes the $ t \times t $ identity matrix, $ 0_{(r-t) \times t} $ is the $ (r-t)\times t $ identically zero matrix, $ B \in \mathbb{F}_q^{t \times (r+\delta -1 - t)} $ and $ C \in \mathbb{F}_q^{(r-t) \times (r+\delta -1 - t)} $.
\end{enumerate}
We finally define a \textit{global code} as follows:
\begin{enumerate}
\item[3.]
\textit{Global code}: Let $ \mathcal{C}_{glob} \subseteq \mathbb{F}_{q^m}^n $ be given by
$$ \mathcal{C}_{glob} = \mathcal{C}_{out} {\rm diag}(\underbrace{D,D, \ldots, D}_{g \textrm{ times}}), $$
where
$$ D = \left( \begin{array}{c|cccc}
I_t & B & B & \ldots & B \\
\hline
0 & C & 0 & \ldots & 0 \\
0 & 0 & C & \ldots & 0 \\
\vdots & \vdots & \vdots & \ddots & \vdots \\
0 & 0 & 0 & \ldots & C
\end{array} \right) \in \mathbb{F}_q^{(t + N(r-t)) \times (t + N(r+\delta - 1 - t))}. $$
\end{enumerate}
\end{construction}

In terms of generator matrices, the previous construction can be easily described as follows. If $ G_{out} \in \mathbb{F}_{q^m}^{k \times n_o} $ is a generator matrix of $ \mathcal{C}_{out} $, then a generator matrix of $ \mathcal{C}_{glob} $ is simply given by
$$ G_{glob} = G_{out} {\rm diag}(\underbrace{D,D, \ldots, D}_{g \textrm{ times}}) \in \mathbb{F}_{q^m}^{k \times n} . $$

We now show that this construction always yields a maximally recoverable code as desired.

\begin{theorem}
The global code $ \mathcal{C}_{glob} \subseteq \mathbb{F}_{q^m}^n $ from Construction \ref{construction 1} is an LRC of type $ (r,\delta, N,t) $ with $ N $-availability that can correct any simple erasure pattern as in Definition \ref{def correctable patterns}. In particular, it is maximally recoverable if $ t=1 $ by Theorem \ref{thm:correctable}. Furthermore, if $ k \leq gt $, it has information $ N $-availability.
\end{theorem}
\begin{proof}
We define the sets $ T_i $ and $ R_{i,j} $ as in (\ref{eq sets T_i R_ij}). Clearly, $ \mathcal{C}_{glob} $ is an LRC of type $ (r,\delta,N,t) $ with $ N $-availability, and if $ k \leq gt $, it has information $ N $-availability. Hence we only need to prove that it can correct any simple erasure pattern as in Definition \ref{def correctable patterns}.

Consider an erasure pattern $ \mathcal{E} \subseteq [n] $ that is locally correctable, as in Definition \ref{def locally correctable erasures}, and of maximum size, i.e., $ |\mathcal{E}| = gN(\delta - 1) $. By Remark \ref{remark simple patterns if MDS}, we only need to show that $ \mathcal{C}_{glob}|_{\overline{\mathcal{E}}} $ is MDS, where $ \overline{\mathcal{E}} = [n] \setminus \mathcal{E} $. 

Without loss of generality (by reordering the sets $ R_{i,1}, \ldots, R_{i,N} $), we may assume that $ |\mathcal{E} \cap R_{i,1}| = \delta - 1 $ and $ |(R_{i,j} \setminus T_i) \cap \mathcal{E} | = \delta - 1 $, for all $ j \in [2,N] $ and all $ i \in [g] $. Therefore the code $ \mathcal{C}|_{\overline{\mathcal{E}}} $ is of the form
$$ \mathcal{C}|_{\overline{\mathcal{E}}} = \mathcal{C}_{out} {\rm diag}(D|_{\overline{\mathcal{E}} \cap R_1}, D|_{\overline{\mathcal{E}} \cap R_2}, \ldots, D|_{\overline{\mathcal{E}} \cap R_g}) \subseteq \mathbb{F}_{q^m}^{n_o}, $$
where we identify $ \overline{\mathcal{E}} \cap R_i $ with $ [t + N(r-t)] $ in the obvious way (they are of the same size). By Lemma \ref{lemma msrd charact mds diag}, we only need to show that $ D|_{\overline{\mathcal{E}} \cap R_i} \in \mathbb{F}_q^{(t + N(r-t)) \times (t + N(r-t))} $ is invertible, for all $ i \in [g] $.

For each $ i \in [g] $, we have that
$$ D|_{\overline{\mathcal{E}} \cap R_i} = \left( \begin{array}{c|cccc}
\multirow{2}{*}{$A|_{\overline{\mathcal{E}} \cap R_{i,1}}$} & B|_{\overline{\mathcal{E}} \cap (R_{i,2}\setminus T_i)} & B|_{\overline{\mathcal{E}} \cap (R_{i,3}\setminus T_i)} & \ldots & B|_{\overline{\mathcal{E}} \cap (R_{i,N}\setminus T_i)} \\
 & 0 & 0 & \ldots & 0 \\
 \hline
 0 & C|_{\overline{\mathcal{E}} \cap (R_{i,2}\setminus T_i)} & 0 & \ldots & 0 \\
 0 & 0 & C|_{\overline{\mathcal{E}} \cap (R_{i,3}\setminus T_i)} & \ldots & 0 \\
 \vdots & \vdots & \vdots & \ddots & \vdots \\
 0 & 0 & 0 & \ldots & C|_{\overline{\mathcal{E}} \cap (R_{i,N}\setminus T_i)}
\end{array} \right). $$ 
Now, since $ A \in \mathbb{F}_q^{r \times (r+\delta -1)} $ generates an $ r $-dimensional MDS code and $ |\mathcal{E} \cap R_{i,1}| =\delta -1 $, then $ A|_{\overline{\mathcal{E}} \cap R_{i,1}} \in \mathbb{F}_q^{r \times r} $ is invertible. For the same reason, the matrix
$$ A|_{[t] \cup (\overline{\mathcal{E}} \cap (R_{i,j}\setminus T_i))} = \left( \begin{array}{c|c}
I_t & B|_{\overline{\mathcal{E}} \cap (R_{i,j}\setminus T_i)} \\
\hline
0_{(r-t) \times t} & C|_{\overline{\mathcal{E}} \cap (R_{i,j}\setminus T_i)}
\end{array} \right) \in \mathbb{F}_q^{r \times r} $$
is also invertible, which in turn implies that $ C|_{\overline{\mathcal{E}} \cap (R_{i,j}\setminus T_i)} \in \mathbb{F}_q^{(r-t) \times (r-t)} $ is invertible, for $ j \in [2,N] $. Therefore, due to its upper block-triangular shape, we conclude that the matrix $ D|_{\overline{\mathcal{E}} \cap R_i} \in \mathbb{F}_q^{(t + N(r-t)) \times (t + N(r-t))} $ is invertible, for all $ i \in [g] $, and we are done.
\end{proof} 

If we choose as outer code a linearized Reed--Solomon code \cite[Def. 31]{linearizedRS} (see Definition \ref{def lrs codes}), then we have the restrictions $ q > g $ and $ m \geq t + N(r-t) $. Furthermore, if we choose as the local code a doubly extended Reed--Solomon code \cite[Ch. 11, Sec. 5]{macwilliamsbook}, then we also have the restriction $ q \geq r+\delta -1 $. Since both of these code families are explicit, we conclude the following.

\begin{corollary} \label{cor mr-lrc construction gen matrix lrs}
For any positive integers $ r $, $ \delta $, $ t $, $ g $, $ k $ and $ N $, with $ t \leq \min \{ \delta - 1 , r \} $ and $ k \leq g(t+N(r-t)) $, there exists an explicit LRC of type $ (r,\delta, N,t) $ with $ N $-availability that can correct any simple erasure pattern as in Definition \ref{def correctable patterns} (in particular, it is maximally recoverable if $ t=1 $ by Theorem \ref{thm:correctable}), of dimension $ k $ (thus number of global parities $ h = g(t+N(r-t))-k $) and with field sizes of the form
$$ q^m = \mathcal{O}(\max \{ g+1, r+\delta - 1 \})^{t + N(r-t)}. $$
Furthermore, if $ k \leq gt $, then the code may be chosen to have information $ N $-availability.
\end{corollary}

A generator matrix of such a code can be explicitly given using linearized Reed--Solomon codes (see Definition \ref{def lrs codes}). Set $ q $ as the smallest prime power at least $ \max \{ g+1, r+\delta - 1 \} $ (thus $ q = \mathcal{O}(\max \{ g+1, r+\delta - 1 \}) $) and $ m = t + N(r-t) $. Let $ \beta_1, \beta_2, \ldots, \beta_m \in \mathbb{F}_{q^m} $ be a basis of $ \mathbb{F}_{q^m} $ over $ \mathbb{F}_q $. Define the matrix $ D \in \mathbb{F}_q^{(t+N(r-t)) \times (t+N(r+\delta -1-t))} $ as in Construction \ref{construction 1}, where $ A \in \mathbb{F}_q^{r \times (r+\delta - 1)} $ generates a doubly extended Reed--Solomon code \cite[Ch. 11, Sec. 5]{macwilliamsbook}. Define
$$ (\gamma_1, \gamma_2, \ldots, \gamma_{t+N(r+\delta-1-t)}) = (\beta_1, \beta_2, \ldots, \beta_m) D \in \mathbb{F}_{q^m}^{t+N(r+\delta-1-t)}. $$
Next, let $ a_1,a_2, \ldots, a_g \in \mathbb{F}_{q^m}^* $ be such that $ N_{\mathbb{F}_{q^m}/\mathbb{F}_q}(a_i) \neq N_{\mathbb{F}_{q^m}/\mathbb{F}_q}(a_j) $ if $ i \neq j $ (where $ N_{\mathbb{F}_{q^m}/\mathbb{F}_q} $ denotes the norm of $ \mathbb{F}_{q^m} $ over $ \mathbb{F}_q $ \cite[Ch. 2, Sec. 3]{lidl}). Finally, let $ \mathcal{C}_{out} \subseteq \mathbb{F}_{q^m}^{gr} $ be a $ k $-dimensional linearized Reed--Solomon code with generator matrix $ G_{out} \in \mathbb{F}_{q^m}^{k \times (gr)} $ over $ (\mathbf{a},\boldsymbol\beta) $, as in Definition \ref{def lrs codes}. Following Construction \ref{construction 1}, a generator matrix of the MR-LRC from Corollary \ref{cor mr-lrc construction gen matrix lrs} is
$$ G_{glob} = (G_1|G_2| \ldots |G_g) \in \mathbb{F}_{q^m}^{k \times n}, $$
where
$$ G_i = \left( \begin{array}{cccc}
\gamma_1 & \gamma_2 & \ldots & \gamma_{t+N(r+\delta-1-t)} \\
 \gamma_1^q a_i & \gamma_2^q a_i & \ldots & \gamma_{t+N(r+\delta-1-t)}^q a_i \\
 \gamma_1^{q^2} a_i^{q+1} & \gamma_2^{q^2} a_i^{q+1} & \ldots & \gamma_{t+N(r+\delta-1-t)}^{q^2} a_i^{q+1} \\
\vdots & \vdots & \ddots & \vdots \\
 \gamma_1^{q^{k-1}} a_i^{\frac{q^{k-1}-1}{q-1}} & \gamma_2^{q^{k-1}} a_i^{\frac{q^{k-1}-1}{q-1}} & \ldots & \gamma_{t+N(r+\delta-1-t)}^{q^{k-1}} a_i^{\frac{q^{k-1}-1}{q-1}}
\end{array} \right) \in \mathbb{F}_{q^m}^{k \times (t + N(r+\delta-1-t))}, $$
for $ i \in [g] $.

\section{A First Construction Using Parity-Check Matrices} \label{sec const parity-check}

In this section, we present a second construction, which is based on parity-check matrices of MSRD codes (see Section \ref{appendix}). The construction extends the constructions in \cite[Sec. III]{cai-field} and \cite[Sec. III-A]{gopi-field}, which are essentially equivalent but appeared independently. When $ N = 1 $, this construction achieves the optimal field-size order in Corollary \ref{cor lower bounds field h <= g} when $ n = \Theta(r^2) = \Theta(g^2) $, and $ h $ and $ \delta $ considered as constants (see \cite{cai-field, gopi-field}).

As in Section \ref{sec const generator}, we will present a general construction based on an outer MSRD code and local MDS codes. Later, we will particularize this construction using an outer linearized Reed--Solomon code and local Reed--Solomon codes, which guarantees small field sizes.

\begin{construction} \label{construction 2}
Fix positive integers $ r $, $ \delta $, $ t $, $ g $ and $ N $, with $ t \leq \min \{ \delta - 1 , r \} $ and $ h \leq r $. Define the code lengths
$$ n_o = ghN \quad \textrm{and} \quad n = g(t + N(r+\delta - 1-t)). $$
Next choose:
\begin{enumerate}
\item
\textit{Outer code}: An $ h $-dimensional $ \mathbb{F}_{q^m} $-linear code $ \mathcal{C}_{out} \subseteq \mathbb{F}_{q^m}^{n_o} $ that is MSRD over $ \mathbb{F}_q $ for the length partition $ (g, hN) $. See Definition \ref{def msrd}.
\item
\textit{Local codes}: An $ (r+\delta-1, h+\delta-1) $ MDS code $ \mathcal{C}_{loc} \subseteq \mathbb{F}_q^{r + \delta - 1} $, linear over $ \mathbb{F}_q $. Let $ A $ be a generator matrix of $ \mathcal{C}_{loc} $ of the form 
$$ A = \left( \begin{array}{cc}
I_t & B \\
0_{(\delta - 1 -t) \times t} & C \\
0_{h \times t} & D
\end{array} \right) \in \mathbb{F}_q^{(h+\delta-1) \times (r+\delta-1)}, $$ 
where $ B \in \mathbb{F}_q^{t \times (r+\delta -1 - t)} $, $ C \in \mathbb{F}_q^{(\delta - 1 -t) \times (r+\delta -1 - t)} $ and $ D \in \mathbb{F}_q^{h \times (r+\delta -1 - t)} $, and such that the linear code generated by
$$ A^\prime = \left( \begin{array}{cc}
I_t & B \\
0_{(\delta - 1 -t) \times t} & C 
\end{array} \right) \in \mathbb{F}_q^{(\delta-1) \times (r+\delta-1)} $$
is also MDS. Such a generator matrix $ A $ exists for any MDS code (let $ A $ be systematic in the first $ h+\delta-1 $ coordinates and apply \cite[Ch. 11, Th. 8]{macwilliamsbook}).
\end{enumerate}
We finally define a \textit{global code} as follows:
\begin{enumerate}
\item[3.]
\textit{Global code}: Let $ \mathcal{C}_{glob} \subseteq \mathbb{F}_{q^m}^n $ be the $ \mathbb{F}_{q^m} $-linear code with parity-check matrix
$$ H = \left( \begin{array}{cccc}
P & 0 & \ldots & 0 \\
0 & P & \ldots & 0 \\
\vdots & \vdots & \ddots & \vdots \\
0 & 0 & \ldots & P \\
G_1Q & G_2Q & \ldots & G_gQ
\end{array} \right) \in \mathbb{F}_{q^m}^{(gN(\delta - 1) + h) \times n}, $$
where $ (G_1|G_2|\ldots|G_g) \in \mathbb{F}_{q^m}^{h \times (ghN)} $ is a generator matrix of $ \mathcal{C}_{out} $,
$$ P = \left( \begin{array}{ccccc}
I_t & B & 0 & \ldots & 0 \\
0 & C & 0 & \ldots & 0 \\
I_t & 0 & B & \ldots & 0 \\
0 & 0 & C & \ldots & 0 \\
\vdots & \vdots & \vdots & \ddots & \vdots \\
I_t & 0 & 0 & \ldots & B \\
0 & 0 & 0 & \ldots & C
\end{array} \right) \in \mathbb{F}_q^{(N(\delta - 1)) \times (t + N(r+\delta - 1 - t))} $$
and
$$ Q = \left( \begin{array}{ccccc}
0_{h \times t} & D & 0 & \ldots & 0 \\
0_{h \times t} & 0 & D & \ldots & 0 \\
\vdots & \vdots & \vdots & \ddots & \vdots \\
0_{h \times t} & 0 & 0 & \ldots & D 
\end{array} \right) \in \mathbb{F}_q^{(hN) \times (t + N(r+\delta - 1 - t))}. $$
\end{enumerate}

\end{construction}

We next show that this construction yields a maximally recoverable code as desired. 

\begin{theorem} \label{th const 2 is MR-LRC}
If $ h \leq r $, then the global code $ \mathcal{C}_{glob} \subseteq \mathbb{F}_{q^m}^n $ from Construction \ref{construction 1} is an LRC of type $ (r,\delta, N,t) $ with $ N $-availability that can correct any simple erasure pattern as in Definition \ref{def correctable patterns}. In particular, it is maximally recoverable if $ t=1 $ by Theorem \ref{thm:correctable}. Furthermore, if $ k \leq gt $, then the code may be chosen to have information $ N $-availability. 
\end{theorem}
\begin{proof}
We define the sets $ T_i $ and $ R_{i,j} $ as in (\ref{eq sets T_i R_ij}). Due to the shape of $ H $ and $ P $ in Construction \ref{construction 2}, we observe that a codeword $ \mathbf{c} \in \mathcal{C}_{glob} $ restricted to the coordinates in $ R_{i,j} $ satisfies the equations
$$ \mathbf{c}|_{R_{i,j}} \left( \begin{array}{cc}
I_t & B \\
0 & C
\end{array} \right)^\intercal = \mathbf{0}, $$
i.e., $ \mathbf{c}|_{R_{i,j}} (A^\prime)^\intercal = \mathbf{0} $. This means that $ \mathbf{c}|_{R_{i,j}} $ belongs to the $ (r+\delta-1,r) $ MDS code with parity-check matrix $ A^\prime $, for $ j \in [N] $ and $ i \in [g] $. This proves that $ \mathcal{C}_{glob} $ is an LRC of type $ (r,\delta,N,t) $ with $ N $-availability, and if $ k \leq gt $, it has information $ N $-availability. 

Now we prove that it can correct any simple erasure pattern $ \mathcal{E} \subseteq [n] $ as in Definition \ref{def correctable patterns}. That is, $ \mathcal{E} = \mathcal{E}_1 \cup \mathcal{E}_2 $, where the union is disjoint, $ \mathcal{E}_1 $ is locally correctable (Definition \ref{def locally correctable erasures}) and $ |\mathcal{E}_2| \leq h $. Without loss of generality, we may assume that $ \mathcal{E} $ is of maximum size and that $ \mathcal{E}_1 $ can be locally repaired starting from the first recovery sets, i.e., 
\begin{enumerate}
\item
$ | \mathcal{E}_1 \cap R_{i,1} | = \delta - 1 $, for all $ i \in [g] $,
\item
$ | \mathcal{E}_1 \cap (R_{i,j} \setminus T_i) | = \delta - 1 $, for all $ j \in [2,N] $ and all $ i \in [g] $, and
\item
$ | \mathcal{E}_2 | = h $.
\end{enumerate}
In order to be able to correct such erasures, we need the matrix $ H|_\mathcal{E} \in \mathbb{F}_{q^m}^{(gN(\delta - 1) + h) \times |\mathcal{E}|} $ to be invertible (notice that it is square, i.e., $ |\mathcal{E}| = gN(\delta - 1) + h $).

Consider sets $ \overline{\Delta}_{i,j} \subseteq \mathcal{E}_1 $ and $ \Delta_{i,j} \subseteq \mathcal{E}_2 $ (thus disjoint) such that $ |\overline{\Delta}_{i,j}| = \delta - 1 $ and
$$ \Delta_{i,j} \cup \overline{\Delta}_{i,j} = \left\lbrace \begin{array}{ll}
\mathcal{E} \cap R_{i,1} & \textrm{ if } j = 1, \\
\mathcal{E} \cap (R_{i,j} \setminus T_i) & \textrm{ otherwise,}
\end{array} \right. $$
for $ j \in [N] $ and $ i \in [g] $. We also set, for $ i \in [g] $,
$$ \Delta_i = \bigcup_{j=1}^N \Delta_{i,j} \quad \textrm{and} \quad \overline{\Delta}_i = \bigcup_{j=1}^N \overline{\Delta}_{i,j}. $$
By the MDS property, the restricted matrices
\begin{equation}
\left( \begin{array}{cc}
I_t & B \\
0 & C
\end{array} \right)_{\overline{\Delta}_{i,1}}, \left( \begin{array}{c}
B \\
C
\end{array} \right)_{\overline{\Delta}_{i,j}} \in \mathbb{F}_q^{(\delta-1)\times(\delta-1)}
\label{eq local matrices for parity proof}
\end{equation}
are invertible, for $ j \in [2,N] $ and $ i \in [g] $. Therefore, $ P_{\overline{\Delta}_i} \in \mathbb{F}_q^{N(\delta-1) \times N(\delta-1)} $ is also invertible due to its triangular shape.

Using the same reasoning, we have that
$$ \left( \begin{array}{cc}
I_t & B \\
0 & C \\
0 & D
\end{array} \right)_{\overline{\Delta}_{i,1} \cup \Delta_{i,1}}, \left( \begin{array}{c}
B \\
C \\
D
\end{array} \right)_{\overline{\Delta}_{i,j} \cup \Delta_{i,j}} \in \mathbb{F}_q^{(h+\delta-1)\times|\overline{\Delta}_{i,j} \cup \Delta_{i,j}|} \textrm{ if } j \geq 2, $$
have full column rank $ |\overline{\Delta}_{i,j} \cup \Delta_{i,j}| \leq h + \delta - 1 $, for $ j \in [N] $ and $ i \in [g] $. Therefore, 
$$ \left( \begin{array}{c}
P_{\overline{\Delta}_i \cup \Delta_i} \\
Q_{\overline{\Delta}_i \cup \Delta_i}
\end{array} \right) \in \mathbb{F}_q^{(N(h + \delta-1)) \times |\overline{\Delta}_i \cup \Delta_i |} $$ 
is of full column rank $ |\overline{\Delta}_i \cup \Delta_i | \leq h + N(\delta-1) $, since reordering its rows we obtain the matrix
$$ \left( \begin{array}{cccc}
 \left( \begin{array}{cc}
 I_t & B \\
 0 & C \\
 0 & D
 \end{array} \right)_{\overline{\Delta}_{i,1} \cup \Delta_{i,1}} & 0 & \ldots & 0 \\
 \left( \begin{array}{cc}
 I_t & 0 \\
 0 & 0 \\
 0 & 0
 \end{array} \right)_{\overline{\Delta}_{i,1} \cup \Delta_{i,1}} &  \left( \begin{array}{c}
 B \\
 C \\
 D
 \end{array} \right)_{\overline{\Delta}_{i,2} \cup \Delta_{i,2}} & \ldots & 0 \\
 \vdots & \vdots & \ddots & \vdots \\
 \left( \begin{array}{cc}
 I_t & 0 \\
 0 & 0 \\
 0 & 0
 \end{array} \right)_{\overline{\Delta}_{i,1} \cup \Delta_{i,1}} & 0 & \ldots & \left( \begin{array}{c}
 B \\
 C \\
 D
 \end{array} \right)_{\overline{\Delta}_{i,N} \cup \Delta_{i,N}}
\end{array} \right), $$
which is of full column rank $ |\overline{\Delta}_i \cup \Delta_i | \leq h + N(\delta-1) $ due to its triangular shape as before.

Now, since the matrices in (\ref{eq local matrices for parity proof}) are invertible, then using elementary $ \mathbb{F}_q $-linear column operations, we may construct invertible matrices $ E_{i,j} \in \mathbb{F}_q^{|\overline{\Delta}_{i,j} \cup \Delta_{i,j}| \times |\overline{\Delta}_{i,j} \cup \Delta_{i,j}|} $ such that
$$ \left( \begin{array}{cc}
I_t & B \\
0 & C \\
0 & D
\end{array} \right)_{\overline{\Delta}_{i,1} \cup \Delta_{i,1}} \cdot E_{i,1} = \left( \begin{array}{cc}
 \left( \begin{array}{cc}
I_t & B \\
0 & C
\end{array} \right)_{\overline{\Delta}_{i,1}} & 0 \\
\left( \begin{array}{cc}
0 & D
\end{array} \right) _{\overline{\Delta}_{i,1}} & W_{i,1}  
\end{array} \right) \textrm{ and} $$
$$ \left( \begin{array}{c}
 B \\
C \\
D
\end{array} \right)_{\overline{\Delta}_{i,j} \cup \Delta_{i,j}} \cdot E_{i,j} = \left( \begin{array}{cc}
 \left( \begin{array}{c}
B \\
C
\end{array} \right)_{\overline{\Delta}_{i,j}} & 0 \\
D_{\overline{\Delta}_{i,j}} & W_{i,j}  
\end{array} \right), \textrm{ for } j \geq 2, $$
for some $ W_{i,j} \in \mathbb{F}_q^{h \times |\Delta_{i,j}|} $ for $ j \in [N] $ and $ i \in [g] $. More concretely, using a suitable ordering of $ \overline{\Delta}_{i,j} \cup \Delta_{i,j} $, the matrices $ E_{i,j} $ are given by
$$ E_{i,1} = \left( \begin{array}{cc}
I_{\delta-1} & - \left( \begin{array}{cc}
I_t & B \\
0 & C
\end{array} \right)_{\overline{\Delta}_{i,1}}^{-1} \left( \begin{array}{cc}
I_t & B \\
0 & C
\end{array} \right)_{\Delta_{i,1}} \\
0 & I_{|\Delta_{i,1}|}
\end{array} \right) , $$
and similarly for $ E_{i,j} $ (simply remove the columns corresponding to $ I_t $), for $ j \in [N] $ and $ i \in [g] $.

Hence, using a suitable ordering of $ \overline{\Delta}_i \cup \Delta_i $, we have
$$ \left( \begin{array}{cc}
P_{\overline{\Delta}_i \cup \Delta_i} \\
Q_{\overline{\Delta}_i \cup \Delta_i}
\end{array} \right) \cdot \left( \begin{array}{cccc}
E_{i,1} & 0 & \ldots & 0 \\
0 & E_{i,2} & \ldots & 0 \\
\vdots & \vdots & \ddots & \vdots \\
0 & 0 & \ldots & E_{i,N}
\end{array} \right) = $$
\begin{equation}
\left( \begin{array}{cc|cc|c|cc}
\left( \begin{array}{cc}
I_t & B \\
0 & C
\end{array} \right)_{\overline{\Delta}_{i,1}} & 0 & 0 & 0 & \ldots & 0 & 0 \\
\hline
\left( \begin{array}{cc}
I_t & 0 \\
0 & 0
\end{array} \right)_{\overline{\Delta}_{i,1}} & F_i & \left( \begin{array}{c}
B \\
C
\end{array} \right)_{\overline{\Delta}_{i,2}} & 0 & \ldots & 0 & 0 \\
\hline
\vdots & \vdots & \vdots & \vdots & \ddots & \vdots & \vdots \\
\hline
\left( \begin{array}{cc}
I_t & 0 \\
0 & 0
\end{array} \right)_{\overline{\Delta}_{i,1}} & F_i & 0 & 0 & \ldots & \left( \begin{array}{c}
B \\
C
\end{array} \right)_{\overline{\Delta}_{i,N}} & 0 \\
\hline
 \left( \begin{array}{cc}
 0 & D
\end{array} \right)_{\overline{\Delta}_{i,1}} & W_{i,1} & 0 & 0 & \ldots & 0 & 0 \\
0 & 0 & D_{\overline{\Delta}_{i,2}} & W_{i,2} & \ldots & 0 & 0 \\
\vdots & \vdots & \vdots & \vdots & \ddots & \vdots & \vdots \\
0 & 0 & 0 & 0 & \ldots & D_{\overline{\Delta}_{i,N}} & W_{i,N}
\end{array} \right),
\label{eq big matrix}
\end{equation}
for certain (possibly nonzero) matrices 
\begin{equation}
F_i = \left( \begin{array}{cc}
I_t & 0 \\
0 & 0 
\end{array} \right)_{\Delta_{i,1}} - \left( \begin{array}{cc}
I_t & 0 \\
0 & 0 
\end{array} \right)_{\overline{\Delta}_{i,1}} \left( \begin{array}{cc}
I_t & B \\
0 & C 
\end{array} \right)_{\overline{\Delta}_{i,1}}^{-1} \left( \begin{array}{cc}
I_t & B \\
0 & C 
\end{array} \right)_{\Delta_{i,1}} \in \mathbb{F}_q^{(\delta - 1) \times |\Delta_{i,1}|} ,
\label{eq def F_{i,j}}
\end{equation}
for $ i \in [g] $. Since $ \left( \begin{array}{cc}
P_{\overline{\Delta}_i \cup \Delta_i} \\
Q_{\overline{\Delta}_i \cup \Delta_i}
\end{array} \right) $ has full column rank $ |\overline{\Delta}_i \cup \Delta_i | $ (which is its number of columns) and multiplying by $ {\rm diag}(E_{i,1},E_{i,2},$ $\ldots, E_{i,N}) $ on the right is an invertible linear transformation of such columns, then the matrix in (\ref{eq big matrix}) is also of full column rank $ |\overline{\Delta}_i \cup \Delta_i | $. 

Next, since the matrices $ \left( \begin{array}{c}
B \\
C
\end{array} \right)_{\overline{\Delta}_{i,j}} $ are invertible, for $ j \in [2,N] $, by further elementary $ \mathbb{F}_q $-linear column operations, we may construct invertible matrices $ M_i \in \mathbb{F}_q^{|\overline{\Delta}_i \cup \Delta_i| \times |\overline{\Delta}_i \cup \Delta_i|} $ such that
$$ \left( \begin{array}{cc}
P_{\overline{\Delta}_i \cup \Delta_i} \\
Q_{\overline{\Delta}_i \cup \Delta_i}
\end{array} \right) \cdot \left( \begin{array}{cccc}
E_{i,1} & 0 & \ldots & 0 \\
0 & E_{i,2} & \ldots & 0 \\
\vdots & \vdots & \ddots & \vdots \\
0 & 0 & \ldots & E_{i,N}
\end{array} \right) \cdot M_i = $$
\begin{equation}
\left( \begin{array}{cc|cc|c|cc}
\left( \begin{array}{cc}
I_t & B \\
0 & C
\end{array} \right)_{\overline{\Delta}_{i,1}} & 0 & 0 & 0 & \ldots & 0 & 0 \\
\hline
\left( \begin{array}{cc}
I_t & 0 \\
0 & 0
\end{array} \right)_{\overline{\Delta}_{i,1}} & 0 & \left( \begin{array}{c}
B \\
C
\end{array} \right)_{\overline{\Delta}_{i,2}} & 0 & \ldots & 0 & 0 \\
\hline
\vdots & \vdots & \vdots & \vdots & \ddots & \vdots & \vdots \\
\hline
\left( \begin{array}{cc}
I_t & 0 \\
0 & 0
\end{array} \right)_{\overline{\Delta}_{i,1}} & 0 & 0 & 0 & \ldots & \left( \begin{array}{c}
B \\
C
\end{array} \right)_{\overline{\Delta}_{i,N}} & 0 \\
\hline
 \left( \begin{array}{cc}
 0 & D
\end{array} \right)_{\overline{\Delta}_{i,1}} & W_{i,1} & 0 & 0 & \ldots & 0 & 0 \\
0 & V_{i,2} & D_{\overline{\Delta}_{i,2}} & W_{i,2} & \ldots & 0 & 0 \\
\vdots & \vdots & \vdots & \vdots & \ddots & \vdots & \vdots \\
0 & V_{i,N} & 0 & 0 & \ldots & D_{\overline{\Delta}_{i,N}} & W_{i,N}
\end{array} \right),
\label{eq big matrix 2}
\end{equation}
for some matrices $ V_{i,2}, \ldots, V_{i,N} \in \mathbb{F}_q^{h \times |\Delta_{i,1}|} $, for $ i \in [g] $. Since the matrices in (\ref{eq local matrices for parity proof}) are invertible and those in (\ref{eq big matrix 2}) have full column rank $ |\overline{\Delta}_i \cup \Delta_i| $, we conclude that
$$ W_i = \left( \begin{array}{ccccc}
 W_{i,1} & 0 & 0 & \ldots & 0 \\
 V_{i,2} & W_{i,2} & 0 & \ldots & 0 \\
 V_{i,3} & 0 & W_{i,3} & \ldots & 0 \\
 \vdots & \vdots & \vdots & \ddots & \vdots \\
 V_{i,N} & 0 & 0 & \ldots & W_{i,N}
\end{array} \right) \in \mathbb{F}_q^{(hN) \times |\Delta_i|} $$
is of full column rank $ |\Delta_i| $, for $ i \in [g] $.

Finally, using a suitable ordering of $ \overline{\Delta}_i \cup \Delta_i $ and defining the invertible matrices $ N_i = {\rm diag}(E_{i,1},E_{i,2}, \ldots, E_{i,N}) M_i \in \mathbb{F}_q^{|\overline{\Delta}_i \cup \Delta_i| \times |\overline{\Delta}_i \cup \Delta_i|} $, we have 
$$ H|_\mathcal{E} \cdot {\rm diag}(N_1, N_2, \ldots, N_g) = $$
$$ \left( \begin{array}{ccccccc}
P_{\overline{\Delta}_1} & 0 & 0 & 0 & \ldots & 0 & 0 \\
0 & 0 & P_{\overline{\Delta}_2} & 0 & \ldots & 0 & 0 \\
\vdots & \vdots & \vdots & \vdots & \ddots & \vdots & \vdots \\
0 & 0 & 0 & 0 &\ldots & P_{\overline{\Delta}_g} & 0 \\
G_1 Q_{\overline{\Delta}_1} & G_1W_1 & G_2 Q_{\overline{\Delta}_2} & G_2W_2 & \ldots & G_g Q_{\overline{\Delta}_g} & G_gW_g
\end{array} \right). $$
Since the matrices $ P_{\overline{\Delta}_i} $ are invertible, we conclude that $ H|_\mathcal{E} $ is invertible if, and only if, 
$$ (G_1W_1| G_2W_2| \ldots |G_gW_g) = (G_1|G_2|\ldots|G_g) {\rm diag}(W_1,W_2,\ldots,W_g) \in \mathbb{F}_{q^m}^{h \times h} $$
is invertible, which holds by Lemma \ref{lemma msrd charact mds diag} since $ \mathcal{C}_{out} $ is MSRD for the length partition $ (g,hN) $ and $ W_i \in \mathbb{F}_q^{(hN) \times |\Delta_i|} $ is of full column rank $ |\Delta_i| $, for $ i \in [g] $, and we are done.
\end{proof}

For an explicit construction, we may proceed as in Section \ref{sec const generator}, and obtain the following.

\begin{corollary} \label{cor mr-lrc construction parity matrix lrs}
For any positive integers $ r $, $ \delta $, $ t $, $ g $, $ h $ and $ N $, with $ t \leq \min\{ \delta-1,r \} $ and $ h \leq \min \{ r, g(t+N(r-t)) \} $, there exists an explicit LRC of type $ (r,\delta, N,t) $ with $ N $-availability that can correct any simple erasure pattern as in Definition \ref{def correctable patterns} (in particular, it is maximally recoverable if $ t=1 $ by Theorem \ref{thm:correctable}), of dimension $ k = g(t+N(r-t))-h $, and with field sizes of the form
$$ q^m = \mathcal{O}(\max \{ g+1, r+\delta - 1 \})^{hN}. $$
Furthermore, if $ k \leq gt $, then the code may be chosen to have information $ N $-availability.
\end{corollary}

In a similar way as in Section \ref{sec const parity-check}, a parity-check matrix of such a code can be explicitly given using linearized Reed--Solomon codes (see Definition \ref{def lrs codes}). Set $ q $ as the smallest prime power at least $ \max \{ g+1, r+\delta - 1 \} $ (thus $ q = \mathcal{O}(\max \{ g+1, r+\delta - 1 \}) $) and $ m = hN $. Let $ a_1,a_2, \ldots, a_g \in \mathbb{F}_{q^m}^* $ be such that $ N_{\mathbb{F}_{q^m}/\mathbb{F}_q}(a_i) \neq N_{\mathbb{F}_{q^m}/\mathbb{F}_q}(a_j) $ if $ i \neq j $ (where $ N_{\mathbb{F}_{q^m}/\mathbb{F}_q} $ denotes the norm of $ \mathbb{F}_{q^m} $ over $ \mathbb{F}_q $ \cite[Ch. 2, Sec. 3]{lidl}). Let $ \beta_1, \beta_2, \ldots, \beta_{hN} \in \mathbb{F}_{q^m} $ be a basis of $ \mathbb{F}_{q^m} $ over $ \mathbb{F}_q $. Consider a generator matrix of an $ h $-dimensional linearized Reed--Solomon code with length partition $ (g,hN) $,
$$ (G_1|G_2| \ldots |G_g) \in \mathbb{F}_{q^{hN}}^{h \times (ghN)}, $$
where
$$ G_i = \left( \begin{array}{cccc}
\beta_1 & \beta_2 & \ldots & \beta_{hN} \\
 \beta_1^q a_i & \beta_2^q a_i & \ldots & \beta_{hN}^q a_i \\
 \beta_1^{q^2} a_i^{q+1} & \beta_2^{q^2} a_i^{q+1} & \ldots & \beta_{hN}^{q^2} a_i^{q+1} \\
\vdots & \vdots & \ddots & \vdots \\
 \beta_1^{q^{h-1}} a_i^{\frac{q^{h-1}-1}{q-1}} & \beta_2^{q^{h-1}} a_i^{\frac{q^{h-1}-1}{q-1}} & \ldots & \beta_{hN}^{q^{h-1}} a_i^{\frac{q^{h-1}-1}{q-1}}
\end{array} \right) \in \mathbb{F}_{q^{hN}}^{h \times (hN)}, $$
for $ i \in [g] $, as in Definition \ref{def lrs codes}. Then a parity-check matrix of the code in Corollary \ref{cor mr-lrc construction parity matrix lrs} is as in Construction \ref{construction 2}, Item 3, where $ A \in \mathbb{F}_q^{(h+\delta -1) \times (r+\delta -1)} $ generates a doubly extended Reed--Solomon code \cite[Ch. 11, Sec. 5]{macwilliamsbook} of dimension $ h+\delta-1 $ and length $ r+\delta -1 $.

\section{A Second Construction Using Parity-Check Matrices} \label{sec const parity-check second} 

In this section, we present a third construction, the second one based on parity-check matrices of MSRD codes (see Section \ref{appendix}). This construction is inspired by the maximally recoverable code extensions for general topologies given in \cite[Sec. V-A]{gopalan-MR}. Even though the extension in \cite[Sec. V-A]{gopalan-MR} is given for general topologies, it is only valid for binary local codes. We extend it to any base finite field and then particularize it to the setting of maximal recovery with locality and availability given in this manuscript, providing LRCs with availability that can correct simple erasure patterns as in Definition \ref{def correctable patterns} (thus being maximally recoverable when $ t=1 $ by Theorem \ref{thm:correctable}). Furthermore, we may further reduce the field sizes obtained in \cite[Sec. V-A]{gopalan-MR} by using MSRD codes instead of MRD codes. For small $ \delta $, the construction in this section gives the lowest field size among all three (see Table \ref{table comparisons}), and the exponent in the field size is the closest to the lower bound from Corollary \ref{cor lower bounds field h <= g} as a function of $ h $.

\begin{construction} \label{construction 3}
Fix positive integers $ r $, $ \delta $, $ t $, $ g $ and $ N $, with $ t \leq \min \{ \delta - 1 , r \} $. Define the code length $ n = g(t + N(r+\delta - 1-t)) $. As in Construction \ref{construction 2}, fix an $ (r+\delta-1, \delta-1) $ MDS code $ \mathcal{C}_{loc} \subseteq \mathbb{F}_q^{r + \delta - 1} $, linear over $ \mathbb{F}_q $, with a generator matrix of the form
$$ A = \left( \begin{array}{cc}
I_t & B \\
0_{(\delta - 1 -t) \times t} & C 
\end{array} \right) \in \mathbb{F}_q^{(\delta-1) \times (r+\delta-1)}, $$ 
where $ B \in \mathbb{F}_q^{t \times (r+\delta -1 - t)} $ and $ C \in \mathbb{F}_q^{(\delta - 1 -t) \times (r+\delta -1 - t)} $. 
We may choose such a generator matrix $ A $ for any MDS code for the same reasons as in Construction \ref{construction 2}. Set now
$$ P_0 = \left( \begin{array}{ccccc}
I_t & B & 0 & \ldots & 0 \\
0 & C & 0 & \ldots & 0 \\
I_t & 0 & B & \ldots & 0 \\
0 & 0 & C & \ldots & 0 \\
\vdots & \vdots & \vdots & \ddots & \vdots \\
I_t & 0 & 0 & \ldots & B \\
0 & 0 & 0 & \ldots & C
\end{array} \right) \in \mathbb{F}_q^{(N(\delta-1)) \times (t + N(r+\delta-1-t))}, $$
which represents the $ N $ local codes in $ R_i $ (the same matrix $ P_0 $ for every $ i \in [g] $). Define now
$$ P = \left( \begin{array}{cccc}
P_0 & 0 & \ldots & 0 \\
0 & P_0 & \ldots & 0 \\
\vdots & \vdots & \ddots & \vdots \\
0 & 0 & \ldots & P_0
\end{array} \right) \in \mathbb{F}_q^{(gN(\delta-1)) \times n} , $$
which represents all $ gN $ local codes (they are all the same local code, repeated $ gN $ times). Next, consider a field extension degree $ m $ and an $ \ell $-wise $ \mathbb{F}_q $-linearly independent set $ \mathcal{S} = \{ \beta_1, \beta_2, \ldots, \beta_{n/g} \} \subseteq \mathbb{F}_{q^m} $ of size $ n/g $ (i.e., any subset of $ \mathcal{S} $ of size $ \ell $ or less is $ \mathbb{F}_q $-linearly independent). Let also $ a_1,a_2, \ldots, a_g \in \mathbb{F}_{q^m}^* $ be such that $ N_{\mathbb{F}_{q^m}/\mathbb{F}_q}(a_i) \neq N_{\mathbb{F}_{q^m}/\mathbb{F}_q}(a_j) $ if $ i \neq j $ (where $ N_{\mathbb{F}_{q^m}/\mathbb{F}_q} $ denotes the norm of $ \mathbb{F}_{q^m} $ over $ \mathbb{F}_q $ \cite[Ch. 2, Sec. 3]{lidl}). Finally, define the matrix
$$ P^+ = \left( \begin{array}{c}
P \\
Q
\end{array} \right) = \left( \begin{array}{cccc}
P_0 & 0 & \ldots & 0 \\
0 & P_0 & \ldots & 0 \\
\vdots & \vdots & \ddots & \vdots \\
0 & 0 & \ldots & P_0 \\
Q_1 & Q_2 & \ldots & Q_g
\end{array} \right) \in \mathbb{F}_{q^m}^{(gN(\delta-1) + h) \times n} , $$
where $ Q = (Q_1|Q_2| \ldots |Q_g) \in \mathbb{F}_{q^m}^{h \times n} $ and, for each $ i \in [g] $, we define
$$ Q_i = \left( \begin{array}{cccc}
\beta_1 & \beta_2 & \ldots & \beta_{n/g} \\
 \beta_1^q a_i & \beta_2^q a_i & \ldots & \beta_{n/g}^q a_i \\
 \beta_1^{q^2} a_i^{q+1} & \beta_2^{q^2} a_i^{q+1} & \ldots & \beta_{n/g}^{q^2} a_i^{q+1} \\
\vdots & \vdots & \ddots & \vdots \\
 \beta_1^{q^{h-1}} a_i^{\frac{q^{h-1}-1}{q-1}} & \beta_2^{q^{h-1}} a_i^{\frac{q^{h-1}-1}{q-1}} & \ldots & \beta_{n/g}^{q^{h-1}} a_i^{\frac{q^{h-1}-1}{q-1}}
\end{array} \right) \in \mathbb{F}_{q^m}^{h \times (t + N(r+\delta-1-t))} . $$
We define the \textit{global code} $ \mathcal{C}_{glob} \subseteq \mathbb{F}_{q^m}^n $ as the $ \mathbb{F}_{q^m} $-linear code with parity-check matrix $ P^+ $.
\end{construction}

The following result extends \cite[Th. 33]{gopalan-MR}, which will later allow us to turn the code $ \mathcal{C}_{glob} $ from Construction \ref{construction 3} into an MR-LRC with availability.

\begin{theorem} \label{th correctable pattern for const 3}
The code $ \mathcal{C}_{glob} \subseteq \mathbb{F}_{q^m}^n $ from Construction \ref{construction 3} can correct any erasure pattern $ \mathcal{E} \subseteq [n] $ such that $ |\mathcal{E}| \leq \ell $ and
$$ |\mathcal{E}| - {\rm Rk}(P|_\mathcal{E}) \leq h , $$
where $ {\rm Rk} $ denotes the rank function.
\end{theorem}
\begin{proof}
Let $ \mathcal{E} \subseteq [n] $ be such that $ |\mathcal{E}| \leq \ell $ and $ |\mathcal{E}| - {\rm Rk}(P|_\mathcal{E}) \leq h $. If we denote by
$$ \ker(P^+|_\mathcal{E}) = \left\lbrace \mathbf{c} \in \mathbb{F}_{q^m}^{|\mathcal{E}|} : (P^+|_\mathcal{E}) \mathbf{c}^\intercal = \mathbf{0} \right\rbrace \subseteq \mathbb{F}_{q^m}^{|\mathcal{E}|} $$
the right kernel of $ P^+|_\mathcal{E} \in \mathbb{F}_{q^m}^{|\mathcal{E}|} $, then we need to prove that $ \dim(\ker(P^+|_\mathcal{E})) = 0 $, or equivalently, $ P^+|_\mathcal{E} $ has maximum rank, $ {\rm Rk}(P^+|_\mathcal{E}) = |\mathcal{E}| $. This is because, in that case, there would be a unique solution to a linear system of equations with coefficient matrix $ P^+|_\mathcal{E} $, i.e., a unique way to correct the erasures. 

Fix $ i \in [g] $. Let $ \mathcal{E}_i = \mathcal{E} \cap R_i $ (with notation as in Definition \ref{def LRC availability}), which may be considered as a subset of $ [t+N(r+\delta -1-t)] $ by a translation since $ |R_i| = t+N(r+\delta -1-t) $. Let also $ e_i = {\rm Rk}(P_0|_{\mathcal{E}_i}) \leq |\mathcal{E}_i| $. After possibly reordering the rows of $ P_0|_{\mathcal{E}_i} $, we may assume that its first $ e_i $ rows are linearly independent (this does not change the proof but simplifies the notation). Then there exists an invertible matrix $ A_i \in \mathbb{F}_q^{|\mathcal{E}_i| \times |\mathcal{E}_i|} $ such that
$$ P_0|_{\mathcal{E}_i} \cdot A_i = \left( \begin{array}{c|c}
 P_i & 0_{e_i \times (|\mathcal{E}_i|-e_i)} \\
\hline
 P_i^\prime & P_i^{\prime \prime} \\
\end{array} \right) \in \mathbb{F}_q^{(N(\delta-1)) \times |\mathcal{E}_i|} , $$
for an invertible matrix $ P_i \in \mathbb{F}_q^{e_i \times e_i} $, and for other (possibly not of full rank) matrices $ P_i^\prime \in \mathbb{F}_q^{(N(\delta-1) - e_i) \times e_i} $ and $ P_i^{\prime \prime} \in \mathbb{F}_q^{(N(\delta-1) - e_i) \times (|\mathcal{E}_i|-e_i)} $.

Now if we set $ A = {\rm diag}(A_1, A_2, \ldots, A_g) \in \mathbb{F}_q^{|\mathcal{E}| \times |\mathcal{E}|} $, which is invertible, we have that
$$ P^+|_\mathcal{E} \cdot A = \left( \begin{array}{c}
P|_\mathcal{E} \\
Q|_\mathcal{E}
\end{array} \right) \cdot A = \left( \begin{array}{cc|cc|c|cc}
 P_1 & 0 & 0 & 0 & \ldots & 0 & 0 \\
 P_1^\prime & P_1^{\prime \prime} & 0 & 0 & \ldots & 0 & 0 \\
 \hline
 0 & 0 & P_2 & 0 & \ldots & 0 & 0 \\
 0 & 0 & P_2^\prime & P_2^{\prime \prime} & \ldots & 0 & 0 \\
 \hline
 \vdots & \vdots & \vdots & \vdots & \ddots & \vdots & \vdots \\
 \hline
 0 & 0 & 0 & 0 & \ldots & P_g & 0 \\
 0 & 0 & 0 & 0 & \ldots & P_g^\prime & P_g^{\prime \prime} \\
 \hline
 \multicolumn{2}{c|}{(Q_1|_{\mathcal{E}_1})A_1} & \multicolumn{2}{c|}{(Q_2|_{\mathcal{E}_2})A_2} & \ldots & \multicolumn{2}{c}{(Q_g|_{\mathcal{E}_g})A_g}
\end{array} \right) \in \mathbb{F}_{q^m}^{(gN(\delta-1)+h) \times |\mathcal{E}|} . $$
Now, since $ |\mathcal{E}_i| \leq |\mathcal{E}| \leq \ell $, for all $ i \in [g] $, and $ \mathcal{S} = \{ \beta_1, \beta_2, \ldots, \beta_{n/g} \} \subseteq \mathbb{F}_{q^m} $ is $ \ell $-wise $ \mathbb{F}_q $-linearly independent, then the matrix
$$ Q|_\mathcal{E} = ((Q_1|_{\mathcal{E}_1})|(Q_2|_{\mathcal{E}_2})|\ldots|(Q_g|_{\mathcal{E}_g})) \in \mathbb{F}_{q^m}^{h \times |\mathcal{E}|} $$
generates a linearized Reed--Solomon code (Definition \ref{def lrs codes}), which is MSRD (Definition \ref{def msrd}) by Theorem \ref{th lrs codes are msrd}. As a consequence, by Lemma \ref{lemma msrd charact mds diag}, we deduce that the matrix
$$ Q|_\mathcal{E} \cdot A = ((Q_1|_{\mathcal{E}_1})A_1|(Q_2|_{\mathcal{E}_2})A_2|\ldots|(Q_g|_{\mathcal{E}_g})A_g) \in \mathbb{F}_{q^m}^{h \times |\mathcal{E}|} $$ 
generates an MDS code of dimension $ \min\{h, |\mathcal{E}|\} $. Next, since 
$$ \sum_{i=1}^g (|\mathcal{E}_i| - e_i) = |\mathcal{E}| - {\rm Rk}(P|_\mathcal{E}) \leq h $$
by hypothesis, then any $ |\mathcal{E}| - {\rm Rk}(P|_\mathcal{E}) $ columns of $ Q|_\mathcal{E} \cdot A $ are linearly independent. Therefore, the matrix $ P^+|_\mathcal{E} \cdot A $ contains the following submatrix of full column rank $ |\mathcal{E}| $,
$$ \left( \begin{array}{cc|cc|c|cc}
 P_1 & 0 & 0 & 0 & \ldots & 0 & 0 \\
 0 & 0 & P_2 & 0 & \ldots & 0 & 0 \\
 \vdots & \vdots & \vdots & \vdots & \ddots & \vdots & \vdots \\
 0 & 0 & 0 & 0 & \ldots & P_g & 0 \\
 \hline
 \multicolumn{2}{c|}{(Q_1|_{\mathcal{E}_1})A_1} & \multicolumn{2}{c|}{(Q_2|_{\mathcal{E}_2})A_2} & \ldots & \multicolumn{2}{c}{(Q_g|_{\mathcal{E}_g})A_g}
\end{array} \right) \in \mathbb{F}_{q^m}^{({\rm Rk}(P|_\mathcal{E}) + h) \times |\mathcal{E}|} . $$
In particular, $  {\rm Rk}(P^+|_\mathcal{E} \cdot A) = |\mathcal{E}| $. Since $ A $ is invertible and $ |\mathcal{E}| $ is the number of columns of $ P^+ $, we conclude that $ {\rm Rk}(P^+|_\mathcal{E}) = {\rm Rk}(P^+|_\mathcal{E} \cdot A) = |\mathcal{E}| $, as we wanted to prove.
\end{proof} 

As a consequence, we may prove that Construction \ref{construction 3} yields an MR-LRC with availability as desired.

\begin{corollary} \label{cor mr-lrc construction parity matrix lrs 2}
For any positive integers $ r $, $ \delta $, $ t $, $ g $, $ h $ and $ N $, with $ t \leq \min \{ \delta-1,r \} $ and $ h \leq g(t+N(r-t)) $, there exists an explicit LRC of type $ (r,\delta, N,t) $ with $ N $-availability that can correct any simple erasure pattern as in Definition \ref{def correctable patterns} (in particular, it is maximally recoverable if $ t=1 $ by Theorem \ref{thm:correctable}), of dimension $ k = g(t+N(r-t))-h $, and with field sizes of the form
$$ q^m \leq \mathcal{O}\left( \frac{n}{g}-1 \right)^{gN(\delta-1) + h}, $$
where $ q = \max \{ g+1, r+\delta -1 \} $. Furthermore, if $ k \leq gt $, then the code may be chosen to have information $ N $-availability.
\end{corollary}
\begin{proof}
Let the notation, the code $ \mathcal{C}_{glob} $ and the parameters $ r $, $ \delta $, $ t $, $ g $, $ h $ and $ N $ be as in Construction \ref{construction 3}. We need to specify the matrix $ A $ (in order to specify the matrix $ P $) and the set $ \mathcal{S} = \{ \beta_1, \beta_2, \ldots, \beta_{n/g} \} $ (in order to specify the matrix $ Q $). The elements $ a_1, a_2, \ldots, a_g \in \mathbb{F}_{q^m}^* $ can be chosen arbitrarily such that $ N_{\mathbb{F}_{q^m}/\mathbb{F}_q}(a_i) \neq N_{\mathbb{F}_{q^m}/\mathbb{F}_q}(a_j) $, if $ i \neq j $, which imposes the restriction $ q \geq g+1 $.

First, choose $ A $ as the generator matrix of a doubly extended Reed--Solomon code \cite[Ch. 11, Sec. 5]{macwilliamsbook}, which imposes the restriction $ q \geq r+\delta-1 $. If we choose $ q = \mathcal{O}(\max \{ g+1, r+\delta-1 \}) $, then both restrictions on $ q $ are satisfied.

Now that $ P $ is specified, set $ \ell = gN(\delta-1) + h $ and let $ \widetilde{H} \in \mathbb{F}_{q^s}^{\ell \times (n/g)} $ be the parity-check matrix of another doubly extended Reed--Solomon code, which this time imposes the condition $ q^s \geq n/g - 1 $, i.e., $ s \geq \log_q(n/g -1) $. Set $ m = s \ell $ and let $ \alpha_1, \alpha_2, \ldots, \alpha_m \in \mathbb{F}_{q^m} $ form a basis of $ \mathbb{F}_{q^m} $ over $ \mathbb{F}_q $. If $ H \in \mathbb{F}_q^{(s\ell) \times (n/g)} $ denotes the matrix obtained by expanding every entry of $ \widetilde{H} $ column-wise over $ \mathbb{F}_q $, then $ H $ is the parity-check matrix of an $ \mathbb{F}_q $-linear code of minimum Hamming distance at least $ \ell + 1 $ and length $ n/g $. Define now
$$ (\beta_1, \beta_2, \ldots, \beta_{n/g} ) = (\alpha_1, \alpha_2, \ldots, \alpha_m) \cdot H \in \mathbb{F}_{q^m}^{n/g}. $$
Since any $ \ell $ columns of $ H $ are $ \mathbb{F}_q $-linearly independent (it is the parity-check matrix of a code of minimum hamming distance at least $ \ell+1 $), then $ \mathcal{S} = \{ \beta_1, \beta_2, \ldots, \beta_{n/g} \} $, as defined above, is $ \ell $-wise $ \mathbb{F}_q $-linearly independent. Thus the set $ \mathcal{S} $ is also specified, and so is the matrix $ Q $. Thus $ \mathcal{C}_{glob} $ in Construction \ref{construction 3} is explicit with definitions as above.

Next we prove the erasure-correction capabilities of the code $ \mathcal{C}_{glob} $. Exactly in the same way as in Theorem \ref{th const 2 is MR-LRC}, the code $ \mathcal{C}_{glob} $ is an LRC of type $ (r,\delta,N,t) $ with $ N $-availability, and if $ k \leq gt $, it has information $ N $-availability. 

Now we prove that it can correct any simple erasure pattern $ \mathcal{E} \subseteq [n] $ as in Definition \ref{def correctable patterns}. That is, $ \mathcal{E} = \mathcal{E}_1 \cup \mathcal{E}_2 $, where the union is disjoint, $ \mathcal{E}_1 $ is locally correctable (Definition \ref{def locally correctable erasures}) and $ |\mathcal{E}_2| \leq h $. We may assume without loss of generality that $ \mathcal{E}_1 $ is of maximum size, i.e., $ |\mathcal{E}_1| = gN(\delta-1) $.

Similarly to the proof of Theorem \ref{th const 2 is MR-LRC}, since $ \mathcal{E}_1 $ is locally correctable and by the shape of matrix $ P $ and the fact that $ A $ generates an MDS matrix, we have that
$$ {\rm Rk}(P|_{\mathcal{E}_1}) = gN(\delta-1). $$ 
Therefore we have
$$ |\mathcal{E}| - {\rm Rk}(P|_{\mathcal{E}}) \leq gN(\delta-1)+h - {\rm Rk}(P|_{\mathcal{E}_1}) \leq h. $$
Moreover, since $ |\mathcal{E}_1| = gN(\delta-1) $, $ |\mathcal{E}_2| \leq h $ and $ \ell = gN(\delta-1) + h $, then
$$ |\mathcal{E}| = |\mathcal{E}_1| + |\mathcal{E}_2| \leq gN(\delta-1) + h = \ell. $$
Hence, by Theorem \ref{th correctable pattern for const 3}, the erasure pattern $ \mathcal{E} $ is correctable, and we have proven that $ \mathcal{C}_{glob} $ can correct any simple erasure pattern as in Definition \ref{def correctable patterns}.

Finally, we upper bound the field size $ q^m $. If we choose $ q^s = \mathcal{O}( n/g - 1) $, then 
$$ q^m = q^{s \ell} = \mathcal{O}\left( \frac{n}{g}-1 \right)^\ell = \mathcal{O}\left( \frac{n}{g}-1 \right)^{gN(\delta-1) + h}. $$
\end{proof}

\begin{remark}
Observe that, in the case $ g = 1 $ (all local sets intersect in $ T $) and $ \delta = 2 $ (every local set can correct $ 1 $ erasure locally), then the field sizes of Constructions \ref{construction 2} and \ref{construction 3} are, respectively,
$$ (r+1)^{hN} = \mathcal{O}(r^{hN}) \quad \textrm{and} \quad \mathcal{O}(t-1+N(r-t+1))^{2N + h} = \mathcal{O}(r^{2N+h}), $$ 
considering $ N $, $ t $ and $ h $ bounded, and $ n $ and $ r $ unrestricted. In this parameter regime, the field-size exponent $ 2N+h $ is generally much smaller than $ hN $.
\end{remark}

\section{Lower Bound on the Field Size} \label{sec lower bound}

In this section, we will derive a lower bound on the field size of an MR-LRC of type $ (r,\delta, N,t) $ with $ N $-availability. We will only need the fact that MR-LRCs with availability can correct any simple erasure pattern by Theorem \ref{thm:correctable} (hence the bound applies potentially to a larger family of LRCs with availability). To derive a lower bound on the field size, we need to identify certain MDS subcodes of the local codes under consideration. We will identify these by alluding to Definition \ref{def locally correctable erasures}, where we identified the set of all correctable erasure patterns of the local code. After identifying these subcodes, we apply arguments similar to \cite{gopi} to derive the lower bound on the field size.

We will need the following auxiliary lemma, which is \cite[Lemma 1]{gopi}.

\begin{lemma}[\cite{gopi}] \label{lemma projective}
Let $ \mathbb{P}^d_\mathbb{F} $ denote the projective space of dimension $ d $ over a field $ \mathbb{F} $, and let $ X_1, X_2, \ldots, X_g \subseteq \mathbb{P}^d_\mathbb{F} $ be pairwise disjoint subsets of sizes $ |X_i| = t $, for $ i \in [g] $. If $ g \geq d+1 $ and $ |\mathbb{F}| < (g/d-1)t - 4 $, then there exists a hyperplane $ H $ of $ \mathbb{P}^d_\mathbb{F} $ which intersects $ d+1 $ distinct subsets among $ X_1, X_2, \ldots, X_g $.
\end{lemma}

We will first give a bound when $ a+2 \leq h \leq g $, where $ a = (\delta-1)N $.

\begin{proposition} \label{prop lower bound a+2 <= h}
If $ a = N(\delta-1) $, $ a+2 \leq h \leq g $ and there exists an MR-LRC $ \mathcal{C} \subseteq \mathbb{F}^n $ of type $ (r,\delta,N,t) $ with $ N $-availability and $ h $ heavy parities, then 
\begin{equation}
|\mathbb{F}| \geq t \left( \frac{g}{h-1}-1 \right) \binom{r+\delta -1-t}{\delta-1}^N - 4.
\label{eq lower bound field size}
\end{equation}
\end{proposition}
\begin{proof}
By Definition \ref{def LRC availability}, the code $ \mathcal{C} $ must have a parity-check matrix of the form
$$ H = \left( \begin{array}{cccc}
P_1 & 0 & \ldots & 0 \\
0 & P_2 & \ldots & 0 \\
\vdots & \vdots & \ddots & \vdots \\
0 & 0 & \ldots & P_g \\
Q_1 & Q_2 & \ldots & Q_g 
\end{array} \right) \in \mathbb{F}^{(gN(\delta-1)+h) \times n}, $$
where $ P_i \in \mathbb{F}^{(N(\delta-1)) \times (t + N(r+\delta-1-t))} $, $ Q_i \in \mathbb{F}^{h \times (t + N(r+\delta-1-t))} $, for $ i \in [g] $, the length of $ \mathcal{C} $ is $ n = g(t + N(r+\delta -1-t)) $ and its dimension is $ k = n - gN(\delta-1) - h $.

Consider the family $ \mathcal{S} $ of subsets $ S \subseteq [t + N(r+\delta-1-t)] $ of size $ a+1 $ such that, for all $ s \in S $, the erasure pattern $ S \setminus \{ s \} $ is locally correctable in each block $ R_i $, for $ i \in [g] $, which means that there exists $ j \in [N] $ such that
\begin{enumerate}
\item
$ |(S\setminus \{s\}) \cap R_{1,j}| = \delta - 1 $, and
\item
$ |(R_{1,\ell} \setminus T_1) \cap (S\setminus \{s\})| = \delta - 1 $, for all $ \ell \in [N] \setminus \{ j \} $.
\end{enumerate}
(See (\ref{eq sets T_i R_ij}) and Definition \ref{def locally correctable erasures}, and notice that $ S $ may be translated to the subset $ R_i $, for all $ i \in [g] $, and it satisfies the same properties as above for every $ R_i $.) In particular, since $ S\setminus \{s\} $ is locally correctable in the block $ R_i $, we have that $ P_i|_{S\setminus \{s\}} \in \mathbb{F}^{a \times a} $ is invertible, for all $ s \in S $. In other words, any $ a \times a $ submatrix of $ P_i|_S \in \mathbb{F}^{a \times (a+1)} $ is invertible, which means that the code it generates is MDS.

The remainder of the proof will follow the steps of \cite[Prop. 1]{gopi} (the case $ N = 1 $). However, in that case, $ \mathcal{S} $ was formed by any subset of size $ a+1 $, whereas this is far from the case in this proof. Later we will need to lower bound the size $ |\mathcal{S}| $, but we skip this step for now.

For each $ S \in \mathcal{S} $, let $ P_{i,S}^\perp \in \mathbb{F}^{(a+1) \times 1} $ be a nonzero column matrix such that $ P_i|_S \cdot P_{i,S}^\perp = 0 $. Note that $ P_{i,S}^\perp $ is unique up to multiplication by a scalar in $ \mathbb{F} $, since $ P_i|_S \in \mathbb{F}^{a \times (a+1)} $ is of full rank $ a $. Furthermore, every entry of $ P_{i,S}^\perp $ is nonzero since its transpose generates a one-dimensional MDS code (the dual of an MDS code is again MDS \cite[Ch. 11, Sec. 2, Th. 2]{macwilliamsbook}). Define also
$$ p_{i,S} = Q_i|_S \cdot P_{i,S}^\perp \in \mathbb{F}^{h \times 1}, $$
for $ i \in [g] $ and $ S \in \mathcal{S} $. We have the following:
\begin{claim} \label{claim 1}
For distinct indices $ \ell_1, \ldots, \ell_h \in [g] $ (here we use that $ h \leq g $) and any subsets $ S_1, \ldots, S_h \subseteq \mathcal{S} $, the matrix $ (p_{\ell_1,S_1}, \ldots, p_{\ell_h,S_h}) \in \mathbb{F}^{h \times h} $ is invertible.
\end{claim}
\begin{proof}
If we define the matrices
\begin{equation*}
\begin{split}
M_1 & = \left( \begin{array}{cccc}
P_{\ell_1}|_{S_1} & 0 & \ldots & 0 \\
0 & P_{\ell_2}|_{S_2} & \ldots & 0 \\
\vdots & \vdots & \ddots & \vdots \\
0 & 0 & \ldots & P_{\ell_h}|_{S_h} \\
Q_{\ell_1}|_{S_1} & Q_{\ell_2}|_{S_2} & \ldots & Q_{\ell_h}|_{S_h}
\end{array} \right) \in \mathbb{F}^{(a+1)h \times (a+1)h}, \\
M_2 & = \left( \begin{array}{cccc}
P_{\ell_1,S_1}^\perp & 0 & \ldots & 0 \\
0 & P_{\ell_2,S_2}^\perp & \ldots & 0 \\
\vdots & \vdots & \ddots & \vdots \\
0 & 0 & \ldots & P_{\ell_h,S_h}^\perp
\end{array} \right) \in \mathbb{F}^{(a+1)h \times h}, \\
M_3 & = \left( \begin{array}{cccc}
0 & 0 & \ldots & 0 \\
0 & 0 & \ldots & 0 \\
\vdots & \vdots & \ddots & \vdots \\
0 & 0 & \ldots & 0 \\
p_{\ell_1,S_1} & p_{\ell_2,S_2} & \ldots & p_{\ell_h,S_h}
\end{array} \right) \in \mathbb{F}^{(a+1)h \times h},
\end{split}
\end{equation*}
then we have that $ M_1M_2 = M_3 $. Since $ \mathcal{C} $ is an MR-LRC and $ S_i $ is formed by a locally correctable pattern in $ R_{\ell_i} $ plus one more element, for $ i \in [h] $, then $ M_1 $ must be invertible by Theorem \ref{thm:correctable}. Since the columns of $ M_2 $ are non-zero and have pairwise disjoint supports, then $ M_2 $ is of full column rank. Thus we deduce that $ M_3 $ is of full column rank too, which means that $ (p_{\ell_1,S_1}, \ldots, p_{\ell_h,S_h}) \in \mathbb{F}^{h \times h} $ is invertible.
\end{proof}

In particular, $ p_{i,S} $ and $ p_{j,T} $ are not scalar multiples of each other for any subsets $ S,T \subseteq \mathcal{S} $ (even if $ S=T $), for distinct $ i,j \in [g] $ (note that this requires $ h \geq 2 $). We next show that the same holds for $ p_{i,S} $ and $ p_{i,T} $ for distinct $ S,T \subseteq \mathcal{S} $.
\begin{claim} \label{claim 2}
For every $ i \in [g] $ and any $ S,T \subseteq \mathcal{S} $ with $ S \neq T $, the vectors $ p_{i,S} $ and $ p_{i,T} $ are not scalar multiples of each other.
\end{claim}
\begin{proof}
Assume that there exists $ \lambda \in \mathbb{F}^* $ such that $ p_{i,S} = \lambda p_{i,T} $. Then
\begin{equation}
\left( \begin{array}{c}
P_i|_S \\
Q_i|_S
\end{array} \right) P_{i,S}^\perp - \lambda \left( \begin{array}{c}
P_i|_T \\
Q_i|_T
\end{array} \right) P_{i,T}^\perp = \left( \begin{array}{c}
0 \\
p_{i,S}
\end{array} \right) - \lambda \left( \begin{array}{c}
0 \\
p_{i,T}
\end{array} \right) = 0.
\label{eq proof field size lin combination}
\end{equation}
Since every entry in $ P_{i,S}^\perp $ and $ P_{i,T}^\perp $ is non-zero and there exists at least one element in $ S \setminus T $ (or in $ T \setminus S $), then (\ref{eq proof field size lin combination}) implies that we have a non-trivial linear dependence of the columns of 
\begin{equation}
\left( \begin{array}{c}
P_i|_{S \cup T} \\
Q_i|_{S \cup T}
\end{array} \right) \in \mathbb{F}^{(a+h) \times |S \cup T|}.
\label{eq proof field size matrix full rank}
\end{equation}
However, since $ a+2 \leq h $, then
$$ | S \cup T | \leq |S|+|T| = 2a+2 \leq a+h. $$
In other words, $ | T \setminus S | \leq h-1 $. Thus $ S \cup T $ is a set of $ a $ coordinates in $ S $ that are locally correctable plus another $ h $ coordinates. Since $ \mathcal{C} $ is an MR-LRC, then $ \mathcal{C} $ can correct the erasure pattern $ S \cup T $ by Theorem \ref{thm:correctable}, and hence the matrix in (\ref{eq proof field size matrix full rank}) must be of full column rank $ |S \cup T| $. Therefore we arrive at a contradiction, and the claim is proven.
\end{proof}

In conclusion, by Claims \ref{claim 1} and \ref{claim 2} (which use the assumptions $ a+2 \leq h \leq g $), no two vectors $ p_{i,S} $ and $ p_{j,T} $, for $ i,j \in [g] $ and $ S,T \subseteq \mathcal{S} $, are scalar multiples of each other, that is, they define different projective points in $ \mathbb{P}^{h-1}_\mathbb{F} $, unless $ (i,S) = (j,T) $. Define now $ X_i = \{ p_{i,S} : S \subseteq \mathcal{S} \} $, for $ i \in [g] $. Then $ |X_i| = |\mathcal{S}| $, for $ i \in [g] $, and $ X_1, \ldots, X_g \subseteq \mathbb{P}^{h-1}_\mathbb{F} $ are pairwise disjoint. By Claim \ref{claim 1}, there is no hyperplane in $ \mathbb{P}^{h-1}_\mathbb{F} $ that contains $ h $ points from distinct subsets among $ X_1, \ldots, X_g $. Therefore, by Lemma \ref{lemma projective}, 
\begin{equation}
|\mathbb{F}| \geq \left( \frac{g}{h-1} - 1 \right) \cdot |\mathcal{S}| - 4.
\label{eq lower bound field size proof 1}
\end{equation}
The only step left is finding a lower bound on $ |\mathcal{S}| $. Consider $ \mathcal{T} $ as the collection of subsets $ S \subseteq [t+N(r+\delta-1-t)] $ formed by one element in $ T_1 $ and any $ \delta-1 $ elements in $ R_{1,j}\setminus T_1 $, for every $ j \in [N] $ (recall the notation from (\ref{eq sets T_i R_ij})). We now check that $ S \in \mathcal{S} $, for all $ S \in \mathcal{T} $. First we observe that $ |S| = N(\delta-1)+1 = a+1 $. Next, let $ s \in S $. If $ s \in T_1 $, then $ S\setminus \{ s \} $ is formed by a disjoint union of subsets of $ R_{1,j}\setminus T_1 $, for $ j \in [N] $, each of size $ \delta-1 $. Thus $ S \setminus \{s \} $ is clearly locally correctable (Definition \ref{def locally correctable erasures}). Now assume that $ s \in R_{1,j} $, for some $ j \in [N] $. Then $ |(S\setminus \{s\}) \cap R_{1,j}| = \delta - 1 $ and, therefore, $ |(R_{1,\ell} \setminus T_1) \cap (S\setminus \{s\})| = \delta - 1 $, for all $ \ell \in [N] \setminus \{ j\} $. Thus $ S \setminus \{s\} $ is also locally correctable. In conclusion, $ S \in \mathcal{S} $. That is, $ \mathcal{T} \subseteq \mathcal{S} $ and
$$ |\mathcal{S}| \geq |\mathcal{T}| \geq t \binom{r+\delta-1-t}{\delta - 1}^N. $$
Combining this lower bound with (\ref{eq lower bound field size proof 1}), we obtain (\ref{eq lower bound field size}), and we are done.
\end{proof}

In the previous bound we used that $ a+2 \leq h \leq g $. We now adapt the bound to the case $ h \leq \min \{ a+1,g \} $.

\begin{proposition} \label{prop lower bound h < a+2}
If $ a = N(\delta-1) $, $ h \leq \min \{ a+1,g \} $ and there exists an MR-LRC $ \mathcal{C} \subseteq \mathbb{F}^n $ of type $ (r,\delta,N,t) $ with $ N $-availability and $ h $ heavy parities, then 
\begin{equation}
|\mathbb{F}| \geq t \left( \frac{g}{h-1}-1 \right) \binom{r + \left\lfloor \frac{h-2}{N} \right\rfloor -t}{\left\lfloor \frac{h-2}{N} \right\rfloor}^N - 4.
\label{eq lower bound field size h < a+2}
\end{equation}
\end{proposition}
\begin{proof}
As done in \cite[Prop. 2]{gopi}, we only need to adjust the proof of Proposition \ref{prop lower bound a+2 <= h}, as follows. Let the notation be as in that proof, set $ \ell = \delta - 1 - \lfloor (h-2)/N \rfloor $, and consider subsets $ S_j \subseteq R_{1,j} \setminus T_1 \subseteq [t+N(r+\delta -1-t)] $ of size $ |S_j| = \ell $, for $ j \in [N] $ (this is possible since $ \ell \leq \delta - 1 \leq |R_{1,j}\setminus T_1| $ by Definition \ref{def LRC availability}). Consider the subfamily $ \mathcal{S}^\prime \subseteq \mathcal{T} $ formed by the sets $ S \in \mathcal{T} $ such that $ S_j \subseteq S $, for all $ j \in [N] $. 

In the proof of Proposition \ref{prop lower bound a+2 <= h}, the only step where we used $ a+2 \leq h $ was in the proof of Claim \ref{claim 2}, to show that $ |S\cup T| \leq 2a+2 \leq a+h $, for any two distinct $ S,T \in \mathcal{S} $. Now we do not have that $ 2a+2 \leq a+h $. However, since
\begin{equation*}
\begin{split}
\bigcup_{j=1}^N S_j & \subseteq S \cap T, \textrm{ and} \\
\left| \bigcup_{j=1}^N S_j \right| & = \ell N \geq N(\delta - 1) - (h-2) = a-h+2 , 
\end{split}
\end{equation*}
then we still have that
$$ |S \cup T| = |S|+|T|-|S\cap T| = 2a+2 - |S\cap T| \leq a+h, $$
for any two distinct $ S,T \subseteq \mathcal{S}^\prime $. The remainder of the proof is exactly as that of Proposition \ref{prop lower bound a+2 <= h}. However, we need to replace the collection $ \mathcal{T} $ by $ \mathcal{S}^\prime $. Since
$$ |\mathcal{S}^\prime| = t \binom{r+\delta -1-t -\ell}{\delta -1 -\ell}^N = t \binom{r + \left\lfloor \frac{h-2}{N} \right\rfloor -t}{\left\lfloor \frac{h-2}{N} \right\rfloor}^N , $$
the bound (\ref{eq lower bound field size h < a+2}) follows.
\end{proof}

Combining the bounds (\ref{eq lower bound field size}) and (\ref{eq lower bound field size h < a+2}), we obtain the following asymptotic lower bound when $ h \leq g $. 

\begin{corollary} \label{cor lower bounds field h <= g}
If $ h \leq g $ and there exists an MR-LRC $ \mathcal{C} \subseteq \mathbb{F}^n $ of type $ (r,\delta,N,t) $ with $ N $-availability and $ h $ heavy parities, then 
\begin{equation}
|\mathbb{F}| = \Omega_{h,\delta,N}\left( gt \cdot r^{\min \{ N(\delta -1) , N \left\lfloor \frac{h-2}{N} \right\rfloor \}} \right) .
\label{eq lower bound field size h <= g}
\end{equation}
\end{corollary}
\begin{proof}
Combine (\ref{eq lower bound field size}) and (\ref{eq lower bound field size h < a+2}), together with the inequality $ \binom{u}{v} \geq (u/v)^v $, for integers $ 0 \leq v \leq u $, and recall that $ h $, $ \delta $ and $ N $ are considered as constants.
\end{proof}

\section*{Acknowledgements}

The authors wish to thank the anonymous reviewers for their very helpful comments, which greatly improved this paper.

This work was done in part while the authors were visiting the Simons Institute for the Theory of Computing (University of California, Berkeley). 

The first author is supported by MCIN/AEI/10.13039/501100011033 and the European Union NextGenerationEU/PRTR (Grant no. TED2021-130358B-I00), and by MICIU/AEI/ 10.13039/501100011033 and ERDF/EU (Grant no. PID2022-138906NB-C21).


\bibliographystyle{plain}

\end{document}